\documentclass[12pt,a4paper]{article}

\usepackage[utf8]{inputenc}
\usepackage[T1]{fontenc}
\usepackage{amsmath,amssymb,amsfonts,amsthm}
\usepackage{geometry}
\usepackage{graphicx}
\usepackage{hyperref}
\hypersetup{
  pdftitle={HFIPay: Privacy-Preserving Identifier-Routed Payments with Verifiable Claim Binding},
  pdfauthor={Jian Sheng Wang},
  pdfsubject={Privacy-preserving identifier-routed cryptocurrency payments},
  pdfkeywords={identifier-based payment, privacy-preserving payment, verifiable quote, blinded claim binding, zero-knowledge authorization, blockchain}
}
\usepackage{url}
\usepackage{booktabs,multirow,tabularx}
\usepackage{algorithm}
\usepackage{algpseudocode}
\usepackage{tikz}
\usetikzlibrary{shapes.geometric, arrows.meta, positioning, fit, calc, backgrounds}
\usepackage{xcolor}
\usepackage{caption}
\usepackage{float}

\newcommand{\keccak}{\mathsf{keccak256}}

\newcommand{\HFIPay}{\textnormal{\textsc{HFIPay}}}
\newcommand{\ACEGF}{\textnormal{\textsc{Ace-Gf}}}
\newcommand{\ZKACE}{\textnormal{\textsc{ZK-ACE}}}
\newcommand{\DIDP}{\textnormal{\textsc{Didp}}}
\newcommand{\IDcom}{\mathsf{ID_{com}}}

\newtheorem{definition}{Definition}

\newtheorem{lemma}{Lemma}[section]
\newtheorem{proposition}{Proposition}[section]

\title{HFIPay: Privacy-Preserving Identifier-Routed Payments\\with Verifiable Claim Binding}

\author{
  \texttt{Jian Sheng Wang} \\
  \texttt{jason@acechain.io} \\
  Montreal, Canada
}
\date{November 24, 2026}

\begin{document}
\maketitle

\begin{abstract}
Human-friendly identifiers such as email addresses and phone numbers are convenient payment targets, but direct mappings from identifiers to blockchain addresses make balances and transaction histories enumerable by anyone who knows the identifier.

We present \HFIPay{}, a relay-assisted protocol for privacy-preserving identifier-routed cryptocurrency payments.
The relay resolves the identifier off-chain and registers only a random intent identifier, a per-intent blinded binding $\rho_i$, and the quoted payment tuple on-chain; no identifier or reusable recipient tag is published before claim.
In a \emph{verified-quote} deployment, the sender verifies an attested quote proving that $\rho_i$ was derived from the same hidden binding handle as the recipient's attested binding-key commitment, preventing relay-side recipient substitution before funding.

Claims are authorized by a zero-knowledge proof, built in the pattern of \ZKACE{}~\cite{zkace}, that the claimant controls the deterministic identity whose epoch-scoped handle opens the blinded binding and authorizes release of the quoted asset and amount to a chosen destination.
We define observer-model games for enumeration resistance and pre-claim unlinkability, state the composition needed for post-quote claim correctness, and characterize relay compromise and post-claim linkability.
The claims are conditioned on relay-directory privacy before compromise, correctness of the binding-attestation layer in the verified-quote mode, and the imported knowledge soundness of the claim proof system; \HFIPay{} does not provide full transaction privacy or amount hiding.

\end{abstract}

\section{Introduction}\label{sec:intro}

\subsection{Motivation}

To send cryptocurrency, users must manage opaque hexadecimal addresses, select the correct network, and hold native tokens for transaction fees.
These requirements make peer-to-peer crypto payments hard for non-technical users and error-prone even for experienced ones.

Human-friendly identifiers such as email addresses, phone numbers, and social media handles (e.g., X/Twitter usernames), on the other hand, are ubiquitous and already anchor most online identity.
If cryptocurrency could be sent to such an identifier, with no wallet setup, address copying, or gas token acquisition, the experience would resemble ordinary digital payments.

\subsection{Privacy Challenge}

A na\"ive implementation maps user identifiers directly to blockchain addresses via a deterministic function (e.g., $\keccak(\texttt{identifier}) \to \texttt{address}$), where the identifier might be an email, phone number, or social handle.
This is simple, but any third party who knows a user's identifier can then derive the user's on-chain address and inspect its transaction history, token balances, and counterparties.

For public figures the exposure is worse: an adversary knowing only \nolinkurl{alice@example.com}, a public phone number, or a well-known X handle could reconstruct Alice's complete financial profile on every supported chain.

\subsection{Our Approach}

\HFIPay{} addresses this by separating \emph{identifier routing}, \emph{recipient binding}, \emph{sender-side quote verification}, and \emph{claim authorization}.
The sender specifies only a human-friendly identifier such as an email address.
The relay privately resolves that identifier to a registered recipient record, samples a random \texttt{intentId}, derives a one-time deposit address, and commits only an intent-specific blinded binding $\rho_i$ together with the quoted payment metadata on-chain.
In the verified-quote deployment, the relay must additionally provide the sender with an attested quote proving that the quoted $\rho_i$ is tied to the intended recipient's hidden binding handle.
The recipient later claims by proving control of the deterministic identity bound to that blinded value.

The key properties are:
\begin{enumerate}
    \item \textbf{Enumeration resistance.} No on-chain data reveals which funded intents belong to a known email, phone number, or social handle.
    \item \textbf{Pre-claim unlinkability.} Multiple pre-claim payments to the same identifier do not expose a reusable public recipient tag.
    \item \textbf{Post-quote claim correctness.} Once a sender has accepted a valid quote and the resulting intent stores the quoted blinded-binding tuple, claim authorization is enforced on-chain by a zero-knowledge proof rather than by relay discretion.
    \item \textbf{Device-portable recovery.} Because the recipient identity is a deterministic function of a recoverable root, the recipient can rederive the same identity across devices from sealed root material, avoiding a public identifier registry.
    \item \textbf{Relayer-assisted operation.} In an integrated client flow, relays may sponsor fees so users are not required to manage native gas tokens directly.
\end{enumerate}

Informally, \HFIPay{} aims to provide the UX of ``send to email'' while keeping identifier resolution private and moving the spending authority check onto the chain.

\subsection{Contributions}

This paper makes the following contributions:

\begin{enumerate}
    \item \textbf{Protocol design.} We define \HFIPay{} as a relay-assisted routing layer with intent-specific blinded bindings and quoted asset tuples committed before funding and an optional sender-verifiable quote, so that recipient substitution can be detected before funds are sent and post-funding claim authorization is enforced on-chain (Section~\ref{sec:protocol}).
    \item \textbf{Authorization composition.} We factor the identity and authorization dependencies into two named assumptions---a deterministic identity derivation primitive (\DIDP{}) and a knowledge-sound claim proof system---so that \HFIPay{}'s claim-correctness arguments are instantiation-agnostic (Section~\ref{sec:model:didp}). We follow the proof pattern of \ZKACE{}~\cite{zkace} in one instantiation, but the analysis depends only on the stated interface, not on any single construction.
    \item \textbf{Privacy analysis.} We formalize enumeration resistance and pre-claim unlinkability against on-chain observers, and explicitly characterize what becomes linkable after claim (Section~\ref{sec:privacy}).
    \item \textbf{Binding modes.} We distinguish a baseline deployment from a verified-quote deployment that uses attested off-chain binding proofs rather than a public identifier registry (Sections~\ref{sec:model} and~\ref{sec:trust}).
    \item \textbf{Implementation and benchmarks.} We outline compatible EVM and Solana realizations and report measurements from the reference implementation, including a directly on-chain-verifiable claim proof ($128$-byte Groth16 proof, ${\approx}241\text{k}$ gas to verify on EVM) alongside a transparent post-quantum-oriented STARK backend, and per-intent on-chain gas costs (Section~\ref{sec:implementation}).
\end{enumerate}

\HFIPay{} builds on the \ACEGF{} identity framework~\cite{acegf} and uses \ZKACE{} for chain-verifiable authorization~\cite{zkace}.
Its routing mechanism is related to stealth-address systems (ERC-5564~\cite{eip5564}) but operates at the application layer with a different trust boundary; a detailed comparison appears in Section~\ref{sec:related}.

\section{System Model}\label{sec:model}

\subsection{Deterministic Identity Derivation Primitive}
\label{sec:model:didp}

The recipient identity in \HFIPay{} is built on a \emph{deterministic identity derivation primitive} (\DIDP{}), treated as a black-box prerequisite supplied by the wallet or account layer rather than proposed here. We use only its interface and assumed security properties, so that the privacy and claim-correctness arguments can state exactly where identity-root assumptions enter. Any framework satisfying this interface may be used; \ACEGF{}~\cite{acegf} is one such instantiation, and the same primitive underlies \ZKACE{}~\cite{zkace}, whose authorization pattern \HFIPay{}'s claim relation follows.

A \DIDP{} provides two operations:
\begin{itemize}
    \item \textbf{Identity reconstruction:} $\mathsf{REV}_B \leftarrow \mathsf{Unseal}(\mathsf{params}, \mathsf{SA}, \mathsf{Cred})$, where $\mathsf{SA}$ is a sealed artifact encoding the root and $\mathsf{Cred}$ is an authorization credential (e.g., a user passphrase). Returns $\mathsf{REV}_B$ or $\bot$.
    \item {\sloppy \textbf{Context-specific derivation:} $\mathsf{Key} \leftarrow \mathsf{Derive}(\mathsf{REV}_B, \mathsf{Ctx})$, deterministic in $(\mathsf{REV}_B, \mathsf{Ctx})$.\par}
\end{itemize}

{\sloppy
\HFIPay{} interacts only with the output of this primitive: the recipient holds $\mathsf{REV}_B$ (obtained via $\mathsf{Unseal}$ or any equivalent account mechanism) and from it derives the identity commitment $\IDcom_B$ and, for each epoch $e$, the binding handle $u_{B,e} = H(\mathsf{Derive}(\mathsf{REV}_B, \mathsf{Ctx}_{\mathsf{bind}} \| e))$. The correctness and usability of wallet-side reconstruction, including cross-device portability (re-running $\mathsf{Unseal}$ on a new device from the sealed artifact), are outside our scope.\par}

We assume the following \DIDP{} properties:
\begin{itemize}
    \item \textbf{Determinism:} a fixed $(\mathsf{REV}_B, \mathsf{Ctx})$ always yields the same derived key, so $\IDcom_B$ and every $u_{B,e}$ are reproducible.
    \item \textbf{Context isolation:} distinct $\mathsf{Ctx} \neq \mathsf{Ctx}'$ yield computationally independent derived keys; in particular, handles for different epochs are independent.
    \item \textbf{Identity-root recovery hardness:} formalized by the game $\mathsf{Game}^{\mathsf{rec\text{-}didp}}_{\mathcal{A}}(\lambda)$:
    \begin{enumerate}
        \item the challenger samples $\mathsf{REV}_B \xleftarrow{\$} \{0,1\}^{\lambda}$;
        \item $\mathcal{A}$ is given oracle access to $\mathcal{O}_{\mathsf{pub}}$, which on $(s, d)$ returns $H(\mathsf{REV}_B \| s \| d)$, and to $\mathcal{O}_{\mathsf{derive}}$, which on $\mathsf{Ctx}$ returns $H(\mathsf{Derive}(\mathsf{REV}_B, \mathsf{Ctx}))$;
        \item $\mathcal{A}$ outputs $\mathsf{REV}_B^\star$ and wins if $\mathsf{REV}_B^\star = \mathsf{REV}_B$.
    \end{enumerate}
    For every polynomial-time $\mathcal{A}$, $\mathsf{Adv}^{\mathsf{rec\text{-}didp}}_{\mathcal{A}}(\lambda) = \Pr[\mathsf{REV}_B^\star = \mathsf{REV}_B]$ is negligible in $\lambda$.
    \item \textbf{Handle pseudorandomness:} root-recovery hardness alone does \emph{not} imply that epoch handles are hidden; indeed, $\mathcal{O}_{\mathsf{derive}}$ returns them on the binding context. The privacy analysis additionally requires that, to a party with \emph{no} oracle or directory access (the on-chain observer), each $u_{B,e}$ is computationally indistinguishable from uniform. This follows from PRF security of $\mathsf{Derive}$ keyed by $\mathsf{REV}_B$ together with the domain-separated hash, and this property, not recovery hardness, is the one Propositions~\ref{prop:enum} and~\ref{prop:unlink} actually invoke.
\end{itemize}

Recovery hardness alone does not underpin claim correctness, because the relay holds $u_{B,e}$ and a claim needs only some root whose epoch handle equals $u_{B,e}$, not $\mathsf{REV}_B$ itself. Claim correctness therefore also requires preimage resistance of the map $\mathsf{REV} \mapsto H(\mathsf{Derive}(\mathsf{REV}, \mathsf{Ctx}_{\mathsf{bind}} \| e))$ at the handle $u_{B,e}$, which holds when $H$ is modeled as a random oracle. The privacy analysis invokes handle pseudorandomness above: an observer to whom $u_{B,e}$ is indistinguishable from uniform cannot recompute the handles that link intents to a recipient.

\paragraph{Claim proof system.}
Beyond the \DIDP{}, \HFIPay{} assumes a non-interactive zero-knowledge proof system for the claim relation $\Phi_{\mathsf{claim}}$ (Section~\ref{sec:protocol}) with two standard properties:
\begin{itemize}
    \item \textbf{Knowledge soundness:} an accepting claim proof implies a witness, so that a successful claim certifies knowledge of a \DIDP{} root whose epoch handle opens the committed blinded binding $\rho$ and authorizes the claim message $m$. Because the adversary observes honest claim proofs, including a pending claim in the mempool, the claim system needs this property in its simulation-extractable form; weak simulation extractability (re-randomization only for the same statement), which Groth16 provides, suffices, since every security-critical claim field is part of the statement and the on-chain state transition rejects a second claim for the same intent.
    \item \textbf{Zero-knowledge:} the proof reveals nothing about $\mathsf{REV}_B$ or the handle $u_{B,e}$ beyond the public inputs.
\end{itemize}
The same two properties are required of the \emph{quote} proof system for $\Phi_{\mathsf{quote}}$ (Section~\ref{sec:intent}); zero-knowledge is essential there, since the quote witness is the epoch handle $u_{B,e}$---a quote transcript that leaked $u_{B,e}$ would let the receiving sender recompute $\rho$ for every intent of that epoch and break enumeration resistance.
This is the classical paradigm of obtaining an unforgeable authorization (a signature in the standard sense) from a proof of knowledge of a committed secret~\cite{bellare-goldwasser,gmr-sig}, with the identity commitment $\IDcom_B$ in the role of the public key. We prove \HFIPay{}'s own relation $\Phi_{\mathsf{claim}}$ in the pattern of \ZKACE{}~\cite{zkace}, which realizes exactly this paradigm; but any proof system meeting the two properties above for $\Phi_{\mathsf{claim}}$ may be substituted. \HFIPay{}'s claim-correctness arguments depend only on these two named properties, not on the specific instantiation: we \emph{import} knowledge soundness as an assumption and do not re-prove the security of any particular claim-proof construction.
Because $\mathsf{Derive}$ appears inside $\Phi_{\mathsf{claim}}$, the deployed \DIDP{} MUST use a \emph{circuit-native} derivation profile, with the same ZK-friendly hash, encoding, and field representation as the claim circuit. A general-purpose KDF profile such as HKDF-SHA256 satisfies the security interface but would have to be evaluated inside the circuit at a cost of tens of thousands of constraints per compression call; the measurements in Section~\ref{sec:benchmarks} assume the circuit-native profile.

\paragraph{Binding-attestation assumption.}
In the verified-quote deployment, \HFIPay{} also assumes an attestation mechanism for the binding layer $\mathcal{I}$.
An accepted attestation $\tau_B$ must be unforgeable for the issuer key, bound to the issuer's stated verification policy, and fresh for its validity horizon $T_{\mathsf{bind}}$.
Sender clients must verify the issuer identity, signature, normalized identifier, binding-key commitment, binding epoch, and freshness window before accepting a quote.
The claim-correctness statements for verified quotes are therefore conditioned on an honest and uncompromised binding issuer for the attested enrollment event, or on an equivalent issuer policy that makes incorrect attestations auditable and outside the cryptographic protocol claim.
Revocation, issuer key rotation, and multi-issuer acceptance policies are deployment parameters, but they must preserve the same signed statement over $(\mathsf{Norm}(e_B), K_{B,e}, e, T_{\mathsf{bind}})$.

\subsection{Participants}

The protocol involves five roles:

\begin{definition}[Sender]
A party $A$ who initiates a payment by specifying a recipient email address $e_B$, an asset identifier $a$, an amount $v$, and a settlement-network identifier $c$.
Here $a$ denotes the canonical source-chain asset descriptor, e.g., native ETH/SOL, an ERC-20 contract address, or an SPL mint.
If refund is enabled, $A$ also authorizes a refund destination $\gamma_A$ using a sender-controlled signing credential or smart-account policy on the relevant chain.
\end{definition}

\begin{definition}[Recipient]
A party $B$ who controls a \DIDP{} root $\mathsf{REV}_B$ (Section~\ref{sec:model:didp}).
From this root, $B$ derives a salted, hiding identity commitment $\IDcom_B$ (Section~\ref{sec:binding-modes}), and for each binding epoch $e$ a private claim-binding handle
\[
    u_{B,e} = H(\mathsf{Derive}(\mathsf{REV}_B, \mathsf{Ctx}_{\mathsf{bind}} \| e)),
\]
together with the corresponding off-chain binding-key commitment
\[
    K_{B,e} = H(\texttt{"hfipay:bind-key"} \| u_{B,e}).
\]
Deployments SHOULD choose $e$ from a coarse public epoch schedule (e.g., a daily or weekly rotation counter) rather than a recipient-specific random label, so that revealing $e$ as intent metadata does not itself create a reusable recipient tag.
\end{definition}

\begin{definition}[Relay Service]
A service $\mathcal{R}$ that maintains a private directory
\[
    \mathsf{Norm}(e) \mapsto (\IDcom, u_e, K_e, e, \texttt{mode}, \texttt{metadata}),
\]
and a private mapping from identifiers to pending intents.
$\mathcal{R}$ pays blockchain transaction fees on behalf of users (relayer-pays model).
\end{definition}

\begin{definition}[Binding Layer]
An optional application-layer or third-party service $\mathcal{I}$ that helps establish the binding between a normalized identifier and the recipient's current binding-key commitment $K_{B,e}$.
In the baseline deployment, this role may be implemented by ordinary email/OIDC onboarding.
In the verified-quote deployment, it issues an attestation
\[
    \tau_B \gets \mathsf{Attest}_{\mathcal{I}}(\texttt{"hfipay:bind-attest"} \| \mathsf{Norm}(e_B) \| K_{B,e} \| e \| T_{\mathsf{bind}}).
\]
$\mathcal{I}$ MUST independently verify that the presenting party controls $\hat{e}_B$ (e.g., by sending its own email OTP or verifying an OIDC token directly) and MUST NOT rely on the relay's assertion of identifier ownership.
Otherwise, a malicious relay could present its own $K$ to $\mathcal{I}$ while claiming it belongs to $\hat{e}_B$, reducing the verified-quote deployment to the baseline trust model.
Independent identifier verification alone does not exclude an OTP-relay man-in-the-middle: if enrollment flows through the relay's application, a malicious relay can submit its own $K'$ to $\mathcal{I}$ while forwarding $\mathcal{I}$'s challenge to the genuine user for confirmation.
The attestation request MUST therefore be bound end-to-end to the enrollee's client: the verification message shown to the user (OTP email or OIDC consent screen) MUST display a fingerprint of the exact $K_{B,e}$ being attested, and $\mathcal{I}$ SHOULD additionally require a proof of knowledge of the handle $u_{B,e}$ underlying $K_{B,e}$, so that attestations are issued only to parties that actually hold the binding secret.
\end{definition}

\begin{definition}[Observer]
Any party $\mathcal{O}$ with read access to public blockchain state but no access to $\mathcal{R}$'s private data.
This includes blockchain analytics services, other users, and the general public.
\end{definition}

\subsection{Enrollment and Binding Modes}\label{sec:binding-modes}

Before receiving funds under the main protocol analyzed in this paper, the recipient completes a one-time enrollment.
Let $\hat{e}_B = \mathsf{Norm}(e_B)$.
$\mathsf{Norm}$ MUST be a fixed, versioned normalization profile agreed by relay, issuer, and clients (Unicode/IDNA handling for email domains, case and dot/plus conventions for local parts, E.164 for phone numbers); identifier re-assignment (recycled email accounts or phone numbers) is inherited from the identifier ecosystem and is part of the enrollment trust assumption, not the cryptographic protocol.
The recipient derives an identity commitment
\[
    \IDcom_B = H(\mathsf{REV}_B \| s_B \| \mathsf{domain}_{\mathsf{id}})
\]
with a private random salt $s_B$, and chooses or refreshes a binding epoch $e$.
This is \HFIPay{}'s own commitment, and \HFIPay{}'s privacy depends on it being hiding. It differs from the native identifier of \ZKACE{}~\cite{zkace}, which is a public unsalted image of the identity root and is therefore linkable.
For that epoch it derives $u_{B,e}$ and $K_{B,e}$ as defined above.
When a deployment records the epoch label in public intent metadata, $e$ SHOULD be chosen from a coarse public schedule (e.g., a day or week counter) rather than as a recipient-unique nonce, so that publishing $e$ does not create a reusable recipient tag.

We consider two deployment modes:
\begin{itemize}
    \item \textbf{Baseline deployment.} The application verifies control of $\hat{e}_B$ once (e.g., by email OTP, magic link, or OIDC login) and stores a private directory entry containing $\hat{e}_B$, $\IDcom_B$, $u_{B,e}$, and epoch metadata.
    This mode gives a simple UX, but both the initial identifier enrollment and the per-intent recipient binding remain relay trust assumptions.
    \item \textbf{Verified-quote deployment.} The binding layer additionally attests to the normalized identifier, the current binding-key commitment, the binding epoch, and a validity horizon.
    The relay stores $(K_{B,e}, \tau_B)$ with the directory entry and later furnishes a sender-verifiable quote showing that each quoted $\rho_i$ is derived from the same hidden handle as the attested $K_{B,e}$.
\end{itemize}

Cross-device recovery does not replace either mode.
Because $\mathsf{REV}_B$ is deterministically derivable (e.g., from a sealed artifact gated by a user passphrase), the recipient can re-derive the same $\IDcom_B$ and the relevant epoch-scoped binding handles on a new device without a public identifier registry.

To reduce retrospective deanonymization if the relay database is later compromised, deployments SHOULD rotate the binding epoch periodically or on demand and retain old $(u_{B,e}, K_{B,e})$ values only until outstanding intents tied to that epoch have settled.
The public intent metadata may still carry the coarse epoch label $e_i$ needed for recovery; the privacy-sensitive secret is $u_{B,e_i}$, not the epoch label itself.

The formal protocol below assumes the \emph{registered-recipient mode}: the relay already has a directory entry for the identifier when the sender initiates payment.
Lazy first-receipt onboarding is a deployment extension, but its initial binding step inherits the trust model of whichever enrollment mode is used.

\subsection{Threat Model}\label{sec:threat}

We consider the following adversarial capabilities:

\begin{itemize}
    \item \textbf{Passive on-chain observer:} $\mathcal{O}$ can read all on-chain state, including contract storage, transaction history, and event logs.
    \item \textbf{Identifier knowledge:} $\mathcal{O}$ may know the email addresses, phone numbers, or social handles of specific users.
    \item \textbf{Directory compromise:} An adversary may compromise the relay's private directory or intent database; we analyze the impact in Section~\ref{sec:relay-compromise}.
    \item \textbf{Public commitment knowledge:} An adversary may know public identity commitments or destination addresses from other contexts and attempt to link them to \HFIPay{} claims.
    \item \textbf{Front-running of submissions:} Any transaction submitted through the mempool may be copied and submitted by another party.
    \item \textbf{Recipient substitution attempt:} A malicious relay may try to quote a blinded binding derived from its own or an accomplice's hidden handle before the sender funds the intent.
\end{itemize}

The protocol does \emph{not} defend against:
\begin{itemize}
    \item Compromise of the recipient's device, passphrase, or underlying deterministic identity root.
    \item Email account takeover or OIDC account takeover during initial enrollment.
    \item Traffic analysis at the network layer (IP correlation).
    \item Collusion between $\mathcal{R}$ and an on-chain observer that reveals the directory contents in real time.
    \item Relay censorship or availability failure.
    \item Binding-issuer unavailability; if $\mathcal{I}$ is unreachable, new registrations and epoch rotations in the verified-quote mode cannot obtain fresh attestations, and the system degrades to the baseline trust model for affected recipients.
    \item Incorrect attestations by a malicious or compromised binding issuer $\mathcal{I}$.
    \item Recipient substitution by a malicious relay in the baseline deployment, which remains a trusted-operator mode.
\end{itemize}

In particular, we focus on two properties:
\begin{enumerate}
    \item privacy against parties with only public-chain visibility and possible knowledge of human-friendly identifiers; and
    \item claim correctness after sender quote acceptance in the verified-quote deployment, i.e., after the sender has checked the quoted binding, the deterministic deposit address, and the registered on-chain intent tuple for an intent.
\end{enumerate}

The baseline deployment retains application-layer assumptions about both one-time identifier enrollment and per-intent recipient binding correctness before funding.
The verified-quote deployment removes the latter assumption: a sender who verifies the attestation $\tau_B$, the quote proof, the locally recomputed deposit address, and the registered on-chain intent tuple cannot be induced to fund an intent whose blinded binding was derived from a different hidden handle or whose quoted asset metadata was silently changed by the relay.
After funding, a relay can still delay or censor, but it cannot redirect the claim without either changing the committed blinded value before funding and failing sender-side verification, or breaking the security of the claim proof system and hash functions.
The liveness guarantee is therefore conditional on eventual transaction inclusion for relay-submitted transactions.
Deployments that require availability should expose a fallback path in which the sender, recipient, or an independent submitter can directly invoke the relevant registration, claim, or refund method when holding the required quote transcript, claim proof, refund authorization, or intent state.
Fee sponsorship is a usability layer and must not be the only path by which a valid terminal state can be submitted.

\subsection{Design Goals}

\begin{enumerate}
    \item \textbf{G1 (Metadata-Conditioned Enumeration Resistance):} Given a target identifier $e_B$, $\mathcal{O}$ cannot determine which on-chain intents belong to that identifier beyond leakage already present in public metadata.
    \item \textbf{G2 (Metadata-Conditioned Pre-Claim Unlinkability):} Given two pre-claim intents generated for the same identifier, $\mathcal{O}$ cannot determine from on-chain data alone that they share a recipient, after conditioning on matched public metadata.
    \item \textbf{G3 (Claim Authorization After Quote Acceptance):} After a sender has accepted a valid quote and an intent stores the resulting tuple $(\rho_i, a_i, v_i, e_i)$, only a party that can produce the corresponding \ZKACE{}-style claim proof can claim that intent.
    \item \textbf{G4 (Portable Recovery):} The same recipient identity can be re-derived across devices from sealed root material, without a public identifier directory.
    \item \textbf{G5 (Relayer-Assisted Operation):} In an integrated client flow, neither $A$ nor $B$ is required to manually manage gas tokens or chain-specific addresses.
\end{enumerate}

We formally capture and analyze G1 and G2 through metadata-conditioned observer games.
G3 is obtained by composing the intent-binding logic of \HFIPay{} with the knowledge soundness imported for $\Phi_{\mathsf{claim}}$ (Section~\ref{sec:model:didp}); we therefore state the composition explicitly but do not re-prove the security of the claim proof system from first principles here.
The privacy arguments reduce to hiding of the blinded binding under the domain-separated hash, uniform intent identifiers, and the assumed privacy of the relay directory.
Throughout, let $\lambda$ denote the deployment security parameter and let $\kappa$ denote the output length of the instantiated domain-separated hash.
Our concrete instantiation uses 256-bit \texttt{intentId} values, so $\lambda = 256$ for the intent identifiers; the recovery-hardness and handle terms are bounded instead by the entropy of $\mathsf{REV}_B$, which is 128 bits of source entropy in the \ACEGF{} wallet encoding profile and 256 bits in its general construction.
We write $H$ for an abstract domain-separated hash function and appeal to collision resistance, preimage resistance, or target second-preimage resistance wherever the relevant security argument requires it.
We do \emph{not} claim that \HFIPay{} hides explicitly public metadata such as chain, asset type, amount, coarse epoch label, or timing.
Propositions~\ref{prop:enum} and~\ref{prop:unlink} instead condition on that public metadata; deployment-layer traffic shaping such as batching, standard amount buckets, or delayed notifications is an optional anonymity enhancement rather than a protocol requirement.

\paragraph{G3 authorization experiment.}
For the verified-quote deployment, consider the following claim-authorization experiment.
The challenger samples an honest recipient identity, obtains a valid binding attestation for $(\mathsf{Norm}(e_B), K_{B,e}, e)$, and lets adversary $\mathcal{A}$ choose quote metadata $(c,a,v,T_{\mathsf{exp}})$ and interact with the relay and chain subject to sender-side quote verification.
The sender accepts only if the attestation, quote proof, deterministic deposit address, and registered on-chain tuple agree on a fixed $(\texttt{intentId},\rho,a,v,e)$.
The adversary wins if the chain later accepts a claim for that intent and releases the quoted value to some destination $\beta$ without the accepted claim proof having a witness whose epoch handle both opens the committed $\rho$ and derives the attested binding key $K_{B,e}$.

\begin{proposition}[Claim authorization after quote acceptance (G3)]\label{prop:g3}
In the verified-quote deployment, assume (i) unforgeability of the binding attestation for the issuer key, (ii) knowledge soundness of the quote and claim proof systems (for the claim system, in the simulation-extractable form of Section~\ref{sec:model:didp}), (iii) target second-preimage resistance of $H$ on the structured binding domain (or the random-oracle model), and (iv) the on-chain checks of Section~\ref{sec:claim}: \texttt{FUNDED}-state gating, equality of every security-critical claim field against storage (including expiry), and the atomic one-terminal-claim transition.
Then every polynomial-time $\mathcal{A}$ wins the experiment above with probability negligible in $\min(\lambda,\kappa)$.
\end{proposition}
\begin{proof}[Proof sketch]
By knowledge soundness of the claim system, an accepted claim certifies a witness whose epoch handle $u'$ opens the stored $\rho$; by Lemma~\ref{lem:quote-claim} (knowledge soundness of the quote system plus target second-preimage resistance), $u'$ equals the handle $u$ behind the attested $K_{B,e}$ except with negligible probability; by attestation unforgeability, $K_{B,e}$ is the commitment enrolled for $\mathsf{Norm}(e_B)$; and by the on-chain checks, the release parameters $(a,v,\beta,T_{\mathsf{exp}},n)$ are exactly those bound by the claim message and stored at registration. Any winning strategy therefore breaks one of (i)--(iii) or recovers the recipient root.
\end{proof}

\section{Protocol}\label{sec:protocol}

\subsection{Overview}

The registered-recipient protocol proceeds in five phases---\emph{Recipient Setup}, \emph{Quote Generation and Verification}, \emph{Funding}, \emph{Claim}, and optional \emph{Refund}.
Figure~\ref{fig:overview} omits the one-time setup phase for clarity.
Let $d_{\mathsf{dep}}$ denote a deployment-domain tag unique to a given \HFIPay{} deployment.
Let $\mathsf{asset}(c,a)$ denote the canonical on-chain encoding of an asset on chain $c$.
We write $H$ for the abstract protocol hash over a canonical byte encoding of the listed fields.
Concrete implementations may instantiate $H$ with a chain-native primitive such as \texttt{keccak256} on EVM, but they MUST preserve the exact field order and the protocol's domain-separation labels.
Authorization-domain labels use the \texttt{hfipay:*} namespace; implementation-local PDA or storage seeds may use different byte strings because they are not replay-sensitive authorization messages.
Implementations MUST ensure that the byte encoding used by the off-chain client, the on-chain contract, and the ZK circuit are identical for every authorization message; Table~\ref{tab:fields} lists the canonical field order.

\begin{table}[H]
\centering
\caption{Canonical field order for authorization messages.}
\label{tab:fields}
\begin{tabularx}{\textwidth}{l>{\raggedright\arraybackslash}X}
\toprule
\textbf{Message} & \textbf{Canonical field order (left to right)} \\
\midrule
Claim & \texttt{tag}, $d_{\mathsf{dep}}$, $c$, $a$, $e$, \texttt{intentId}, $\rho$, $v$, $\beta$, $T_{\mathsf{exp}}$, $n$ \\
Refund & \texttt{tag}, $d_{\mathsf{dep}}$, $c$, $a$, \texttt{intentId}, $\rho$, $v$, $\gamma_A$, $T_{\mathsf{exp}}$ \\
\bottomrule
\end{tabularx}
\end{table}

\noindent Here $c$ denotes the chain identifier in a deployment-specific encoding (e.g., EIP-155 \texttt{chainId} on EVM, a genesis-hash prefix on Solana), $a$ denotes the asset descriptor, and $e$ denotes the coarse epoch label.
Each field is serialized as a fixed-width or length-prefixed byte string; the exact encoding is deployment-specific but MUST be agreed upon by all three stacks (client, contract, circuit).

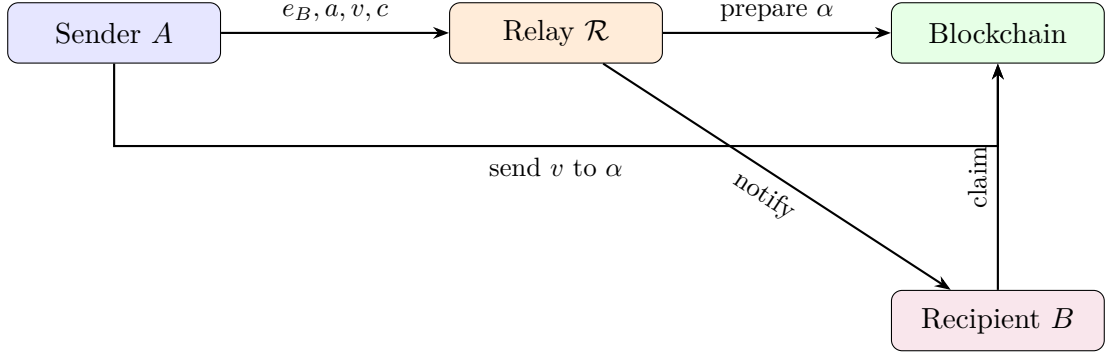
\begin{figure}[H]
\centering
\begin{tikzpicture}[
    node distance=1.2cm and 2.5cm,
    box/.style={draw, rounded corners, minimum width=2.8cm, minimum height=0.8cm, align=center, font=\small},
    arrow/.style={-{Stealth[length=6pt]}, thick},
    lbl/.style={font=\footnotesize, midway, above, sloped}
]
    \node[box, fill=blue!10] (sender) {Sender $A$};
    \node[box, fill=orange!15, right=3cm of sender] (relay) {Relay $\mathcal{R}$};
    \node[box, fill=green!10, right=3cm of relay] (chain) {Blockchain};
    \node[box, fill=purple!10, below=3cm of chain] (recip) {Recipient $B$};

    \draw[arrow] (sender) -- (relay) node[lbl] {$e_B, a, v, c$};
    \draw[arrow] (relay) -- (chain) node[lbl] {prepare $\alpha$};
    \draw[arrow] (sender) |- ++(0,-1.5) -| (chain) node[font=\footnotesize, pos=0.25, below] {send $v$ to $\alpha$};
    \draw[arrow] (relay) -- (recip) node[lbl, below, sloped] {notify};
    \draw[arrow] (recip) -- (chain) node[lbl] {claim};
\end{tikzpicture}
\caption{High-level protocol flow.}
\label{fig:overview}
\end{figure}

Figure~\ref{fig:overview} is intentionally simplified: the relay-to-chain ``prepare $\alpha$'' step subsumes quote generation, any pre-funding intent registration, and the sender's verification of the registered on-chain tuple before funds are transferred.

\subsection{Recipient Setup and Binding}\label{sec:setup}

This phase is performed once per recipient identity and then refreshed whenever the binding epoch changes.

\begin{enumerate}
    \item Recipient $B$ loads or reconstructs $\mathsf{REV}_B$.
    In integrated deployments, this may come from a locally sealed artifact or be recovered on a new device via a deterministic, passphrase-gated recovery procedure.
    \item $B$ derives an identity commitment $\IDcom_B$ and selects a current binding epoch $e$.
    \item $B$ derives the epoch-scoped binding handle $u_{B,e}$ and the corresponding binding-key commitment $K_{B,e}$.
    \item In the baseline mode, the application verifies control of $\mathsf{Norm}(e_B)$ and inserts the corresponding private directory entry.
    \item In the verified-quote mode, $B$ also obtains an attestation $\tau_B$ over the normalized identifier, the binding-key commitment, the binding epoch, and a validity horizon.
    The relay accepts or refreshes the directory entry only if that attestation is valid and current.
    \item The relay retains the current binding epoch and, if necessary, any unsettled prior epochs required to complete claims for already-funded intents.
\end{enumerate}

The sender does not need to know $\IDcom_B$ or $u_{B,e}$.
In the verified-quote mode, the sender sees only $(K_{B,e}, \tau_B)$ and a quote proof that links a quoted $\rho_i$ to the same hidden handle.

\subsection{Quote Generation and Intent Creation}\label{sec:intent}

When sender $A$ initiates a payment to email $e_B$ for asset $a$ and amount $v$ on chain $c$, the steps below describe processing a single intent.
We write $\rho$ when referring to the locally computed value and $\rho_i$ when listing it as an element of an indexed tuple (e.g., a quote or observer view); the two denote the same value for the intent under construction.

\begin{enumerate}
    \item $A$ submits $(e_B, a, v, c, \gamma_A, T_{\mathsf{exp}})$ to relay $\mathcal{R}$ via an authenticated API call, where $\gamma_A$ is an optional refund destination and $T_{\mathsf{exp}}$ is an expiration time.
    \item $\mathcal{R}$ looks up the private directory entry for $\mathsf{Norm}(e_B)$ and obtains the current binding tuple $(\IDcom_B, u_{B,e}, K_{B,e}, e)$ together with optional metadata such as $\tau_B$.
    \item $\mathcal{R}$ generates a cryptographically random intent identifier:
    \[
        \texttt{intentId} \xleftarrow{\$} \{0,1\}^{256}.
    \]
    \item $\mathcal{R}$ computes the deterministic payment address:
    \[
        \alpha = \mathsf{DeriveAddress}(c, \texttt{intentId}).
    \]
    \item $\mathcal{R}$ computes an intent-specific blinded recipient binding:
    \[
        \rho \gets H(\texttt{"hfipay:bind"} \| u_{B,e} \| \texttt{intentId}).
    \]
    \item In the verified-quote mode, $\mathcal{R}$ constructs an off-chain quote proof $\pi_i^{\mathsf{quote}}$ for a relation $\Phi_{\mathsf{quote}}$.
    The quote-verification transcript is public to the sender but not posted on-chain.
    Its public inputs are
    \[
        (\mathsf{Norm}(e_B), \rho_i, K_{B,e_i}, e_i, \texttt{intentId}_i, T_{\mathsf{bind},i}, \tau_B),
    \]
    and its witness is the hidden binding handle $u_{B,e_i}$.
    The relation checks
    \[
    \begin{aligned}
        K_{B,e_i} &= H(\texttt{"hfipay:bind-key"} \| u_{B,e_i}), \\
        \rho_i &= H(\texttt{"hfipay:bind"} \| u_{B,e_i} \| \texttt{intentId}_i).
    \end{aligned}
    \]
    Attestation validity is deliberately \emph{not} part of the ZK relation: $\tau_B$ is public to the sender, who verifies it directly against $(\mathsf{Norm}(e_B), K_{B,e_i}, e_i, T_{\mathsf{bind},i})$ in step~8 below. Folding $\mathsf{VerifyAttest}$ into the circuit would re-introduce in-circuit signature verification, the cost the claim layer is designed to avoid, for no security benefit, since both the attestation and its message fields are already public to the quote's verifier. The quote proof's only job is to show that the \emph{same} hidden handle $u_{B,e_i}$ opens both the attested binding-key commitment and the quoted blinded binding, and it MUST be zero-knowledge (Section~\ref{sec:model:didp}): its witness is $u_{B,e_i}$, and leaking it would break enumeration resistance for the entire epoch.
    The relay or directory service holds this witness and generates $\pi_i^{\mathsf{quote}}$ for the sender; the recipient need not be online during quote creation.
    Here $\tau_B$ attests the identifier-to-$K_{B,e}$ binding, while $\pi_i^{\mathsf{quote}}$ proves that the quoted $\rho_i$ is derived from that same hidden handle.
    As stated in the threat model, this verified-quote mode still relies on the honesty of the binding issuer and registration flow that produced $\tau_B$.
    \item $\mathcal{R}$ returns a quote
    \[
        q_i = (\texttt{intentId}_i, \alpha_i, \rho_i, a_i, v_i, c_i, e_i, \gamma_{A,i}, T_{\mathsf{exp},i}, K_{B,e_i}, T_{\mathsf{bind},i}, \tau_B, \pi_i^{\mathsf{quote}}, T_{\mathsf{quote}}),
    \]
    where $T_{\mathsf{quote}}$ is a quote-expiration time.
    \item In the verified-quote mode, $A$ verifies $\tau_B$ against $(\mathsf{Norm}(e_B), K_{B,e_i}, e_i, T_{\mathsf{bind},i})$, verifies $\pi_i^{\mathsf{quote}}$, checks quote freshness, locally recomputes $\alpha_i = \mathsf{DeriveAddress}(c_i, \texttt{intentId}_i)$, and confirms that the quote fields $(a_i, v_i, c_i, \gamma_{A,i}, T_{\mathsf{exp},i})$ match the sender's requested payment.
    In the baseline mode, $A$ trusts $\mathcal{R}$ for the correctness of $\rho_i$.
    \item If refund is enabled, $A$'s client signs a refund authorization over
    \[
        \sigma_A^{\mathsf{ref}} \gets \mathsf{Auth}_{A}\big(\texttt{"hfipay:refund"} \| d_{\mathsf{dep}} \| c \| a \| \texttt{intentId} \| \rho \| v \| \gamma_A \| T_{\mathsf{exp}}\big),
    \]
    using the sender-controlled credential or smart-account policy associated with $\gamma_A$.
    Let $h^{\mathsf{ref}} = H(\sigma_A^{\mathsf{ref}})$ if refund is enabled, and $h^{\mathsf{ref}} = \bot$ otherwise.
    The full authorization $\sigma_A^{\mathsf{ref}}$ is retained by the sender's client and optionally escrowed with $\mathcal{R}$; only its commitment $h^{\mathsf{ref}}$ is stored on-chain.
    At refund time, the submitting party reveals $\sigma_A^{\mathsf{ref}}$ and the on-chain program reconstructs the expected refund message from the stored intent fields $(c, a, \texttt{intentId}, \rho, v, \gamma_A, T_{\mathsf{exp}})$, verifies $\sigma_A^{\mathsf{ref}}$ against $\gamma_A$'s authorization policy over that message, and checks $H(\sigma_A^{\mathsf{ref}}) = h^{\mathsf{ref}}$.
    If the sender's client and $\mathcal{R}$ both lose $\sigma_A^{\mathsf{ref}}$, the refund path becomes unavailable; the intent then remains in the \texttt{EXPIRED} state until a deployment-specific recovery policy (if any) applies.
    \item $\mathcal{R}$ registers the on-chain intent metadata $(\texttt{intentId}, \rho, a, v, e, T_{\mathsf{exp}}, \gamma_A, h^{\mathsf{ref}})$ or its chain-specific equivalent.
    \item In deployments where registration occurs as a separate transaction, $A$ MUST confirm before funding that the readable on-chain tuple for \texttt{intentId} exactly matches $(\rho_i, a_i, v_i, e_i, T_{\mathsf{exp},i}, \gamma_{A,i}, h^{\mathsf{ref}})$.
    Equivalent atomic create-and-fund flows are also acceptable provided the same tuple is bound in the transaction that transfers funds.
    \item $\mathcal{R}$ stores the private record $(e_B, \texttt{intentId}, \rho, a, v, c, e, \gamma_A, T_{\mathsf{exp}}, t)$ and returns the finalized funding instructions to $A$.
\end{enumerate}

In the verified-quote mode, the sender verifies the quoted binding and full quoted intent tuple \emph{before funding}, and the chain commits to $(\rho, a, v, e, T_{\mathsf{exp}}, \gamma_A, h^{\mathsf{ref}})$ \emph{before or at funding time}.
After that point, the relay cannot substitute a different recipient identity for the same funded intent, swap the quoted asset, or alter expiry/refund semantics without either failing sender-side verification or changing on-chain state.

The ``authenticated API call'' above is an operational requirement, not a specific cryptographic mechanism.
Deployments SHOULD require sender authentication or fee-backed quote creation, rate-limit quote issuance, garbage-collect unfunded intents, and defer recipient notifications until funding has been observed (or until a non-refundable relay fee has been paid) in order to limit spam and notification abuse.

The address derivation function $\mathsf{DeriveAddress}$ is chain-specific:

\begin{definition}[EVM Address Derivation]
For EVM-compatible chains:
\[
    \alpha_{\text{EVM}} = \mathsf{CREATE2}(\texttt{factory}, \texttt{intentId}, \keccak(\texttt{initcode}))
\]
where \texttt{factory} is the deployed \HFIPay{} factory contract and \texttt{initcode} is the EIP-1167 minimal proxy creation code.
The factory also stores the per-intent metadata for $\rho$, asset descriptor $a$, amount $v$, epoch label $e$, expiry, and refund information.
\end{definition}

\begin{definition}[Solana Address Derivation]
For Solana:
\[
    \alpha_{\text{SOL}} = \mathsf{FindPDA}(\texttt{programId}, [\texttt{"hfipay"} \| \texttt{intentId}])
\]
where $\mathsf{FindPDA}$ is the standard Solana Program Derived Address algorithm.
The corresponding state PDA stores $\rho$, the quoted asset descriptor $a$, the quoted amount $v$, the epoch label $e$, expiry, refund metadata, and any needed vault references.
\end{definition}

We write $\mathsf{Auth}_{A}(m)$ for a chain-appropriate refund authorization by the sender.
Claim authorization, by contrast, is provided by a zero-knowledge proof rather than by a direct signature over the recipient's identifier.

\subsection{Funding}\label{sec:funding}

After accepting a quote, sender $A$ transfers $v$ units of the quoted asset $a$ to address $\alpha$ on chain $c$.

On EVM chains, $\alpha$ can receive funds after the factory has recorded the intent metadata, even if the account contract itself has not yet been deployed, because CREATE2 addresses are deterministic.
On Solana, the relay $\mathcal{R}$ pre-creates the state PDA (paying rent) before $A$'s transfer.
For SPL assets, $\mathcal{R}$ also pre-creates a PDA-owned token vault $\mathcal{V}_M$ for each supported mint $M$.

\paragraph{The \texttt{CREATED} $\to$ \texttt{FUNDED} transition.}
A plain asset transfer to $\alpha$ does not execute registry code, so the status flag cannot flip implicitly.
Deployments MUST realize the transition in one of two ways: (i)~an atomic \texttt{createAndFund} entry point in which registration and funding occur in one transaction, or (ii)~a permissionless \texttt{markFunded(intentId)} method that verifies, against the immutable stored tuple, that the deposit address (or vault) balance covers the quoted $(a, v)$ before setting \texttt{FUNDED}.
In either case \texttt{claim} and \texttt{refund} MUST require \texttt{FUNDED} status recorded on-chain rather than inspecting balances at claim time, so that late, duplicate, or partial funding cannot create ambiguous states.

\begin{algorithm}[H]
\caption{Funding Phase}
\begin{algorithmic}[1]
\Require Sender $A$, address $\alpha$, blinded binding $\rho$, quoted asset $a$, quoted amount $v$, chain $c$, optional SPL mint $M$
\If{$c$ is EVM-compatible}
    \State $\mathcal{R}$ invokes \texttt{registerIntent}$(\texttt{intentId}, \rho, a, v, e, T_{\mathsf{exp}}, \gamma_A, h^{\mathsf{ref}})$
    \State $A$ calls \texttt{transfer}$(v)$ to $\alpha$ \Comment{ETH or ERC-20}
\ElsIf{$c$ is Solana and asset is SOL}
    \State $\mathcal{R}$ invokes \texttt{create\_state}$(\texttt{intentId}, \rho, a, v, e, T_{\mathsf{exp}}, \gamma_A, h^{\mathsf{ref}})$
    \State $A$ transfers $v$ lamports to $\alpha$
\ElsIf{$c$ is Solana and asset is an SPL token $M$}
    \State $\mathcal{R}$ invokes \texttt{create\_state}$(\texttt{intentId}, \rho, a, v, e, T_{\mathsf{exp}}, \gamma_A, h^{\mathsf{ref}})$
    \State $\mathcal{R}$ invokes \texttt{create\_vault}$(\texttt{intentId}, M)$ \Comment{Creates PDA-owned token account $\mathcal{V}_M$}
    \State $A$ transfers $v$ units of $M$ to $\mathcal{V}_M$
\EndIf
\end{algorithmic}
\end{algorithm}

\subsection{Wallet-Native One-Time Protected Deposit Address Rail}\label{sec:one-time-address-rail}

The verified-quote flow above is the high-assurance path: the sender accepts a complete tuple $(\rho_i,a_i,v_i,e_i,T_{\mathsf{exp},i},\gamma_{A,i})$ before funding.
Existing wallet name-resolution APIs, however, often expose only a weaker interface of the form
\[
    \text{human-readable recipient} \longrightarrow \text{address}.
\]
To support native wallet send flows without requiring a companion web application as the primary sender interface, \HFIPay{} can instantiate a second deployment rail based on \emph{one-time protected deposit addresses}.

In this rail, the sender enters a human-friendly identifier such as an email address, phone number, or social handle in the wallet recipient field.
The \HFIPay{} resolver samples a fresh random \texttt{intentId}, derives a one-time address $\alpha=\mathsf{DeriveAddress}(c,\texttt{intentId})$, stores a short-lived private resolver record for the normalized identifier, and returns $\alpha$ to the wallet.
The address is therefore deterministic with respect to the fresh intent identifier, but not deterministic with respect to the human identifier:
\[
    \alpha = \mathsf{DeriveAddress}(c,\texttt{intentId}), \qquad
    \texttt{intentId} \xleftarrow{\$} \{0,1\}^{\lambda}.
\]
The resolver MUST NOT use a construction of the form $\alpha=f(\mathsf{Norm}(e_B))$, because that would create a reusable public recipient address and defeat enumeration resistance.

After the wallet sends funds to $\alpha$, the relay/indexer observes the funding transaction and finalizes the pending resolver record by recording the actual asset, amount, funding sender, funding transaction, claim deadline, and refund deadline.
For a minimal EVM deployment, this rail can be restricted to native ETH on a single chain and to one valid funding transaction per address.
The default refund destination is the funding transaction sender.
Thus, if the recipient does not claim before the configured claim deadline, the protocol can refund the funded amount to the wallet that funded the one-time address.

This rail is weaker than the verified-quote rail because the complete asset tuple may be finalized after funding instead of being sender-verified before funding.
For small native wallet sends it keeps the following properties:
\begin{itemize}
    \item the human identifier is not published on-chain;
    \item each payment uses a fresh one-time address;
    \item the recipient claims through the same identity-binding and claim-authorization layer;
    \item unclaimed funds can be refunded to the funding sender; and
    \item the wallet can expose a native ``send to email'' or ``send to handle'' experience.
\end{itemize}

Deployments MUST treat the rail as a \emph{trusted, low-assurance convenience rail}: the wallet receives only an address and cannot verify the intent tuple, factory code hash, quote proof, or issuer attestation, so a compromised resolver (DNS, server, or insider) can return an attacker-controlled address.
Policy limits are therefore mandatory: short address TTLs, small-value limits, native-asset-only support at launch, rate-limited resolver creation, and manual-review handling for duplicate, late, unsupported-token, or exchange-originated funding.
High-value payments MUST use the verified-quote rail, whose sender verification SHOULD pin the chain ID, factory address and implementation code hash, verifier key and circuit version, asset, amount, refund destination, expiry, and the issuer's normalization and attestation policy.
Large payments, partner integrations, ERC-20 support, and stricter sender-side verification belong on the verified-quote rail.

\subsection{Claim and Refund}\label{sec:claim}

When recipient $B$ learns of a pending intent (e.g., from an email notification or a relay query), the claim proceeds as follows:

\paragraph{Authorization Interface Summary.}
\HFIPay{} uses \ZKACE{} only through a narrow interface.
At claim time, the verifier receives the public tuple already listed in the claim relation together with a proof $\pi_i$.
The claimant's witness contains the deterministic identity root (or the material from which it is rederived), the corresponding opening information for $\IDcom_B$, and the epoch-scoped handle needed to satisfy the blinded-binding relation.
Cross-device recovery is orthogonal here: it only re-derives the same deterministic identity root, and is not involved in producing the quote proof.

\begin{enumerate}
    \item $B$ loads or recovers $\mathsf{REV}_B$, reads the epoch label $e_i$ from the intent metadata (or from a retained quote transcript or notification), and re-derives $\IDcom_B$ together with $u_{B,e_i}$.
    \item $B$ chooses a destination address $\beta_i$ (or an integrated client derives a fresh one locally).
    \item $B$ forms a claim message
    \[
        m_i \gets H(\texttt{"hfipay:claim"} \| d_{\mathsf{dep}} \| c_i \| a_i \| e_i \| \texttt{intentId}_i \| \rho_i \| v_i \| \beta_i \| T_{\mathsf{exp}} \| n_i),
    \]
    where $n_i$ is a replay-protection nonce.
    \item $B$ generates a claim proof $\pi_i$, in the \ZKACE{} pattern, for \HFIPay{}'s relation $\Phi_{\mathsf{claim}}$ whose public inputs include $(d_{\mathsf{dep}}, c_i, \IDcom_B, \rho_i, a_i, e_i, \texttt{intentId}_i, v_i, \beta_i, T_{\mathsf{exp}}, n_i)$, or whose verifier fixes $(d_{\mathsf{dep}}, c_i)$ as public constants, and whose witness proves that:
    \begin{enumerate}
        \item $\IDcom_B$ is a valid registered commitment for the deterministic identity root held by $B$;
        \item $u_{B,e_i} = H(\mathsf{Derive}(\mathsf{REV}_B, \mathsf{Ctx}_{\mathsf{bind}} \| e_i))$;
        \item $\rho_i = H(\texttt{"hfipay:bind"} \| u_{B,e_i} \| \texttt{intentId}_i)$; and
        \item the identity underlying $\IDcom_B$ authorizes $m_i$ under \HFIPay{}'s replay rule stated below.
    \end{enumerate}
    \item Anyone (typically $\mathcal{R}$) may submit the claim transaction
    \[
        \texttt{claim}(\texttt{intentId}_i, \IDcom_B, a_i, v_i, \beta_i, T_{\mathsf{exp}}, n_i, \pi_i).
    \]
    The on-chain program first checks that the intent is in the \texttt{FUNDED} state, reconstructs every security-critical field from storage---it checks that the stored blinded binding, asset descriptor, epoch label, quoted amount, \emph{and expiry} equal the supplied $(\rho_i, a_i, e_i, v_i, T_{\mathsf{exp}})$, and that the current time is before the stored expiry---then verifies the proof, atomically marks the intent \texttt{CLAIMED}, and either releases exactly $v_i$ units of $a_i$ directly to $\beta_i$ or performs an immediately following relayed withdrawal of exactly $v_i$ units of $a_i$ to $\beta_i$.
    No expiry or asset field may be taken from calldata without an equality check against storage: a claim parameter that overrode the stored $T_{\mathsf{exp}}$ would let a claimant race the refund path after expiry. \texttt{CLAIMED} and \texttt{REFUNDED} are mutually exclusive terminal states enforced by the atomic status transition.
    Replaying the same transaction after this transition must fail before proof verification or before any value transfer.
    \item After the release transaction succeeds on-chain and reaches the required confirmation depth, $\mathcal{R}$ marks the intent as claimed in its off-chain database.
\end{enumerate}

Because the proof binds $d_{\mathsf{dep}}$, $c_i$, $\texttt{intentId}_i$, $\rho_i$, $a_i$, $v_i$, $\beta_i$, $T_{\mathsf{exp}}$, and $n_i$ either as explicit public inputs or verifier-fixed public constants, the protocol is front-running resistant: a miner or mempool observer may copy the transaction, but successful verification still routes only the quoted asset and amount to $\beta_i$, not to an attacker-chosen address.
The nonce $n_i$ provides replay-domain separation inside the claim message; it does not replace the on-chain idempotency check that enforces one terminal claim per \texttt{intentId}.
Concretely, \HFIPay{} uses its own replay rule in place of the exact-next-nonce check and consumed replay-value set of \ZKACE{}: the proof exposes a nullifier over the claim message $m_i$, and this nullifier is secondary. The \emph{authoritative} replay check is the \texttt{FUNDED} $\to$ \texttt{CLAIMED} state transition itself, which enforces at most one terminal claim per \texttt{intentId} independently of the proof-layer token.

If a deposit address accidentally receives funds beyond the quoted amount $v_i$, that surplus is outside the intended claim relation and requires a separate recovery policy.
The claim path itself does not silently sweep extra balance to the claimant.

Refund remains sender-authorized.
If an intent expires unclaimed, anyone may submit the pre-authorized refund transaction using $\sigma_A^{\mathsf{ref}}$, and the on-chain program releases exactly the quoted amount $v_i$ of asset $a_i$ to $\gamma_A$.
Any surplus beyond $v_i$ that was accidentally sent to the deposit address is not covered by the refund authorization and requires a separate recovery policy, consistent with the claim-side treatment described above.

\begin{lemma}[Quote-to-Claim Composition]\label{lem:quote-claim}
In the verified-quote deployment, suppose a sender accepts a quote after verifying $(K_{B,e_i}, T_{\mathsf{bind},i}, \tau_B, \pi_i^{\mathsf{quote}})$ for fixed $(\texttt{intentId}_i, \rho_i)$, and later the chain accepts a claim proof $\pi_i$ for the same $(\texttt{intentId}_i, \rho_i, e_i)$.
Under knowledge soundness of the quote and claim proof systems (Section~\ref{sec:model:didp}) and target second-preimage resistance of the domain-separated hash $H$ on structured inputs of the form $(\texttt{tag} \| u \| \texttt{intentId})$---or, equivalently, in the random-oracle model up to the usual query bound---both successful verifications must be witnessing the same epoch-scoped hidden handle $u_{B,e_i}$ except with negligible probability.
Consequently, the successfully claiming identity is exactly the deterministic identity family whose binding-key commitment was attested in the accepted quote, even though $\IDcom_B$ itself need not be disclosed to the sender at quote time.
\end{lemma}

\begin{proof}[Proof sketch]
The quote proof establishes the existence of some $u$ such that
\[
    K_{B,e_i} = H(\texttt{"hfipay:bind-key"} \| u)
    \quad\text{and}\quad
    \rho_i = H(\texttt{"hfipay:bind"} \| u \| \texttt{intentId}_i).
\]
The claim proof establishes the existence of some deterministic identity root $\mathsf{REV}_B$ whose epoch-scoped derivative
\[
    u' = H(\mathsf{Derive}(\mathsf{REV}_B, \mathsf{Ctx}_{\mathsf{bind}} \| e_i))
\]
satisfies the same relation
\[
    \rho_i = H(\texttt{"hfipay:bind"} \| u' \| \texttt{intentId}_i).
\]
For a fixed accepted quote, $\rho_i$ is already committed before the claim.
An adversary that does not know the quoted handle and tries to produce a different handle $u' \neq u$ satisfying
\[
    H(\texttt{"hfipay:bind"} \| u' \| \texttt{intentId}_i)=\rho_i
\]
is solving a target second-preimage problem on this structured domain.
In the random-oracle model, after $q_H$ oracle queries, the probability of finding such a distinct $u'$ is at most $q_H 2^{-\kappa}$, up to the negligible probability of malformed encodings or duplicate queries.
Under a standard hash instantiation, this step is captured by target second-preimage resistance rather than ordinary collision resistance alone.
Except with that negligible probability, the accepted claim is therefore bound to the same attested hidden handle that appeared in the sender-verified quote.
Finally, conditions~(a)--(b) of the claim relation (Section~\ref{sec:claim}) tie $u'$ to a root $\mathsf{REV}_B$ that also opens the registered commitment $\IDcom_B$; handle equality therefore extends to identity-family equality: the claiming root is one whose epoch-$e_i$ derivative equals the attested handle.
We deliberately do \emph{not} strengthen the attestation to sign $\IDcom_B$ directly: $\tau_B$ is shown to every quoted sender, and embedding $\IDcom_B$ in it would let any such sender link the identifier to all on-chain claim events that expose $\IDcom_B$, collapsing the pre-claim privacy this design preserves.
\end{proof}

\begin{algorithm}[H]
\caption{Claim Phase (EVM)}
\begin{algorithmic}[1]
\Require Recipient $B$, intent $(\texttt{intentId}, \rho, a, v, e, c)$, destination $\beta$
\State $\alpha \gets \mathsf{CREATE2}(\texttt{factory}, \texttt{intentId}, \keccak(\texttt{initcode}))$
\State $m \gets \keccak(\texttt{"hfipay:claim"} \| d_{\mathsf{dep}} \| c \| a \| e \| \texttt{intentId} \| \rho \| v \| \beta \| T_{\mathsf{exp}} \| n)$ \Comment{$c$ encodes \texttt{chainId}}
\State $B$ re-derives $\IDcom_B$ and $u_{B,e}$
\State $\pi \gets \ZKACE\texttt{.Prove}(d_{\mathsf{dep}}, c, \IDcom_B, \rho, a, e, \texttt{intentId}, v, \beta, T_{\mathsf{exp}}, n)$
\State $\mathcal{R}$ calls \texttt{factory.claimWithProof}$(\texttt{intentId}, \IDcom_B, a, v, \beta, T_{\mathsf{exp}}, n, \pi)$
\State Factory requires \texttt{status == FUNDED}, checks stored $(\rho, a, e, v, T_{\mathsf{exp}})$ and \texttt{block.timestamp} $< T_{\mathsf{exp}}$, verifies $\pi$, marks \texttt{CLAIMED}, and releases exactly $v$ units of $a$ to $\beta$
\end{algorithmic}
\end{algorithm}

On Solana, the same logical relation is verified against the state PDA:

\begin{algorithm}[H]
\caption{Claim Phase (Solana)}
\begin{algorithmic}[1]
\Require Recipient $B$, intent state PDA containing $(\texttt{intentId}, \rho, a, e, v)$ and verifier-fixed deployment/chain domain
\State $B$ re-derives $\IDcom_B$ and $u_{B,e}$
\State $\pi \gets \ZKACE\texttt{.Prove}(d_{\mathsf{dep}}, c, \IDcom_B, \rho, a, e, \texttt{intentId}, v, \beta, T_{\mathsf{exp}}, n)$
\State $\mathcal{R}$ invokes \texttt{claim}$(\texttt{intentId}, \IDcom_B, a, v, \beta, T_{\mathsf{exp}}, n, \pi)$
\State Program requires \texttt{status == FUNDED}, verifies stored $(\rho, a, e, v, T_{\mathsf{exp}})$ and clock $< T_{\mathsf{exp}}$, marks \texttt{CLAIMED}, and transfers exactly $v$ units of the quoted SOL or SPL asset to $\beta$
\end{algorithmic}
\end{algorithm}

\subsection{Intent Lifecycle}

Each \texttt{intentId} is single-use and progresses through one of two terminal paths:

\[
    \texttt{CREATED} \to \texttt{FUNDED} \to \texttt{CLAIMED},
    \qquad
    \texttt{FUNDED} \to \texttt{EXPIRED} \to \texttt{REFUNDED}
\]

An intent enters \texttt{EXPIRED} only if it remains unclaimed after $T_{\mathsf{exp}}$.
If a refund destination $\gamma_A$ was supplied at intent creation, anyone may submit the pre-authorized refund transaction after expiry using the stored authorization $\sigma_A^{\mathsf{ref}}$, and the on-chain program releases exactly the quoted amount $v$ of the quoted asset to $\gamma_A$.
The relay therefore has no unilateral control over expired funds.

\section{Privacy Analysis}\label{sec:privacy}

\subsection{Privacy Scope and Leakage}

\HFIPay{} provides computational privacy against public-chain observers for identifier-to-intent discovery before claim, conditioned on explicitly public metadata.
It does not provide metadata privacy: chain, asset type, amount, coarse epoch label, timestamp, funding source behavior, and notification timing may all act as statistical fingerprints.
If a payment uses an unusual asset, a highly specific amount, or timing that identifies a narrow anonymity set, the protocol does not prevent linkage by statistical analysis of those public fields.

The pre-claim privacy claims below therefore separate computational unlinkability from metadata leakage.
They state that the blinded binding and one-time deposit address do not add a reusable recipient tag beyond the public metadata distribution.
After claim, the leakage surface changes: a reused destination $\beta$ or a verifier interface that exposes a stable $\IDcom_B$ can link claims belonging to the same recipient family even when the corresponding pre-claim intents were unlinkable.

\subsection{On-Chain Data Exposure}

We enumerate what an observer $\mathcal{O}$ can learn from on-chain data alone.

\begin{table}[H]
\centering
\caption{On-chain data visibility.}
\label{tab:visibility}
\begin{tabularx}{\textwidth}{>{\raggedright\arraybackslash}X c >{\raggedright\arraybackslash}X}
\toprule
\textbf{Data} & \textbf{On-Chain?} & \textbf{Reveals Identity?} \\
\midrule
\texttt{intentId} & Yes & No (random) \\
Blinded binding $\rho_i$ & Yes & No (one-time blinded value) \\
Payment address $\alpha$ & Yes & No (derived from random intent) \\
Asset identifier $a_i$ & Yes & Reveals asset type only \\
Amount $v$ & Yes & N/A \\
Epoch label $e_i$ & Yes & No, if drawn from a coarse public schedule rather than a recipient-unique nonce \\
Claim destination $\beta$ (after claim) & Yes & Address only\footnotemark \\
Identity commitment $\IDcom_B$ (claim-time, if exposed by the verifier interface) & Sometimes & Commitment only \\
Email $e_B$ & \textbf{No} & --- \\
$e_B \to (\IDcom_B, u_{B,e}, K_{B,e})$ directory entry & \textbf{No} & --- \\
Timestamp & Yes & N/A \\
\bottomrule
\end{tabularx}
\end{table}
\footnotetext{Linking $\beta$ or $\IDcom_B$ to a real identity requires external information such as exchange KYC data, onboarding records, or a disclosed binding attestation.}

\subsection{Metadata-Conditioned Enumeration Resistance (G1)}

\begin{definition}[Observer View]
Let $\mathcal{O}$ be an observer with read access to all on-chain state and knowledge of a target email $e_B$.
The view of $\mathcal{O}$ is
\[
    \mathsf{View}_{\mathcal{O}} = \{(\texttt{intentId}_i, \alpha_i, \rho_i, a_i, v_i, e_i, t_i, \text{logs}_i)\}_i,
\]
where each tuple is publicly observable on-chain.
\end{definition}

\begin{definition}[G1 Metadata-Conditioned Challenge]
Fix a public metadata tuple
\[
    m = (c, a, v, e, t, \texttt{status})
\]
and a target identifier $e_B$.
The challenger samples a bit $b \xleftarrow{\$} \{0,1\}$, a fresh \texttt{intentId}$^\star \xleftarrow{\$} \{0,1\}^{\lambda}$, and computes
\[
    \alpha^\star = \mathsf{DeriveAddress}(c, \texttt{intentId}^\star).
\]
If $b = 0$, it uses the target recipient's hidden handle $u_{B,e}$; if $b = 1$, it uses an independently sampled comparison recipient handle $u_{B',e}$ with the same public metadata tuple $m$.
In either case it sets
\[
    \rho^\star = H(\texttt{"hfipay:bind"} \| u \| \texttt{intentId}^\star)
\]
for the chosen hidden handle $u$, gives $(e_B, \texttt{intentId}^\star, \alpha^\star, \rho^\star, m)$ to adversary $\mathcal{O}$, and receives a guess $b'$.
The enumeration advantage is
\[
    \mathsf{Adv}^{\mathsf{enum}}_{\HFIPay{},\mathcal{O}}(\lambda)
    =
    \left| \Pr[b' = b] - \frac{1}{2} \right|.
\]
\end{definition}

\begin{proposition}[Metadata-Conditioned Enumeration Resistance]\label{prop:enum}
Assume that: (i) the relay's private directory is not compromised, (ii) the \DIDP{} satisfies handle pseudorandomness (Section~\ref{sec:model:didp}): to a party without oracle or directory access, every current or retained epoch handle $u_{B,e}$ is computationally indistinguishable from uniform, and (iii) $H$ is instantiated as a domain-separated hash with $\kappa$-bit output, modeled as a random oracle (equivalently, $H(\texttt{"hfipay:bind"} \| u \| \cdot)$ is a PRF keyed by $u$), and each \texttt{intentId} is sampled uniformly from $\{0,1\}^{\lambda}$. Collision or preimage resistance alone is \emph{not} sufficient for this game: it requires indistinguishability of $\rho^\star$ from an independently generated binding, which is a pseudorandomness property.
Then for every polynomial-time observer $\mathcal{O}$ in the G1 experiment above,
\[
    \mathsf{Adv}^{\mathsf{enum}}_{\HFIPay{},\mathcal{O}}(\lambda)
\]
is negligible in $\min(\lambda,\kappa)$ beyond leakage already carried by the fixed public metadata tuple $m$.
\end{proposition}

\begin{proof}[Proof sketch]
The public chain reveals only the random intent identifier, the address derived from it, the public asset and amount tuple, the coarse epoch label, and the blinded binding $\rho_i = H(\texttt{"hfipay:bind"} \| u_{B,e} \| \texttt{intentId}_i)$.
Knowledge of $e_B$ does not reveal any live or retained epoch handle, because the mapping from normalized identifiers to $(\IDcom_B, u_{B,e}, K_{B,e})$ is private to the relay directory and the sender-verifiable quote remains off-chain.
An observer without access to the derivation oracle and with the relay directory kept private sees $u_{B,e}$ only as a pseudorandom value, so it cannot test whether a public tuple belongs to $B$ beyond the advantage against handle pseudorandomness and the random-oracle (or PRF) binding of $H$.
Thus the on-chain tuple is computationally independent of $e_B$ up to leakage carried separately by public metadata such as asset type, amount, coarse epoch label, and timing.
\end{proof}

Compare with the na\"ive scheme where $\alpha = \mathsf{DeriveAddress}(c, \keccak(e_B))$: any party knowing $e_B$ would then compute the recipient address directly.

\subsection{Metadata-Conditioned Unlinkability (G2)}

\begin{definition}[Pre-Claim View]
Let $\mathsf{View}_{\mathcal{O}}^{\mathsf{pre}}$ denote the observer view restricted to funded but unclaimed intents.
Equivalently, $\mathsf{View}_{\mathcal{O}}^{\mathsf{pre}}$ excludes any claim-time disclosures such as destination addresses, identity commitments revealed by the verifier, or post-claim events.
\end{definition}

\begin{definition}[G2 Metadata-Conditioned Challenge]
Fix two public metadata tuples $m_1$ and $m_2$ drawn from a deployment's explicitly public metadata distribution.
The challenger samples a bit $b \xleftarrow{\$} \{0,1\}$.
If $b = 0$, it generates two funded pre-claim tuples for the same recipient identifier $e_B$ under hidden handles consistent with $m_1$ and $m_2$.
If $b = 1$, it generates the first tuple for $e_B$ and the second tuple for an independent recipient identifier $e_{B'}$, again matching the same public metadata tuples $m_1$ and $m_2$.
In both cases all \texttt{intentId} values are sampled independently from $\{0,1\}^{\lambda}$ and the challenger gives the resulting pre-claim tuples to adversary $\mathcal{O}$, who outputs a guess $b'$.
The unlinkability advantage is
\[
    \mathsf{Adv}^{\mathsf{unlink}}_{\HFIPay{},\mathcal{O}}(\lambda)
    =
    \left| \Pr[b' = b] - \frac{1}{2} \right|.
\]
\end{definition}

\begin{proposition}[Metadata-Conditioned Pre-Claim Address Unlinkability]\label{prop:unlink}
Under the same assumptions as Proposition~\ref{prop:enum}, every polynomial-time observer $\mathcal{O}$ in the G2 experiment above has
\[
    \mathsf{Adv}^{\mathsf{unlink}}_{\HFIPay{},\mathcal{O}}(\lambda)
\]
negligible in $\min(\lambda,\kappa)$ once the compared executions are conditioned on the same explicitly public metadata tuples $(m_1, m_2)$.
\end{proposition}

\begin{proof}[Proof sketch]
Each intent uses an independently sampled \texttt{intentId}.
The public address $\alpha_i$ is derived only from that random value, and the blinded binding $\rho_i$ is a hash of the hidden epoch handle $u_{B,e}$ combined with the fresh \texttt{intentId}.
For an observer without $u_{B,e}$, the pair $(\alpha_i, \rho_i)$ is distributed like an independently generated one-time tuple subject only to the matched public metadata.
Hence there is no publicly visible structure that distinguishes ``same underlying identifier twice'' from ``two different identifiers once each'' before claim.
\end{proof}

\begin{proposition}[Post-Claim Leakage Bound]
After claim, repeated use of the same public destination $\beta$ and repeated revelation of the same public commitment $\IDcom_B$ can reveal common control at the blockchain-identity layer.
If the verifier interface exposes $\IDcom_B$, then all claims using that same public commitment are immediately linkable even if their pre-claim intents were not.
In the instantiation presented in Section~\ref{sec:claim}, $\IDcom_B$ \emph{is} an explicit \texttt{claim} argument, so post-claim linkability is total for that instantiation: an observer can cluster all claims of a recipient by commitment, and a single external deanonymization of one destination $\beta$ then reveals the recipient's full claim history.
Pre-claim unlinkability (G2) should therefore be understood as a \emph{temporal} guarantee for this instantiation.
Deployments that care about post-claim privacy SHOULD replace the explicit commitment argument with a set-membership statement---the claim relation proves that $\IDcom_B$ lies in the registered-commitment set via one Merkle-path check inside the circuit, Semaphore-style---or refresh user-facing public commitments on a schedule coordinated with enrollment refreshes; the membership variant removes the stable on-chain tag at modest circuit cost and is the recommended long-term design.
\end{proposition}

\paragraph{Passive-observer scope.}
G1 and G2 model a \emph{passive} on-chain observer.
An active adversary can additionally interact with the system: acting as a sender, it can request quotes for a target identifier---learning whether the identifier is registered at all and, in the verified-quote mode, the epoch commitment $K_{B,e}$ (the cross-sender linkability noted below)---and can correlate its own resolver queries with subsequent on-chain registrations by timing.
Deployments must rate-limit and fee-gate quote issuance (Section~\ref{sec:intent}); formalizing an active-query variant of G1/G2 with a resolver oracle is future work.

\begin{proposition}[Cross-Sender Linkability in Verified-Quote Mode]
In the verified-quote deployment, all senders who receive quotes for the same recipient within the same binding epoch observe the same binding-key commitment $K_{B,e}$ in their respective quotes.
Colluding senders can therefore determine that they are paying the same recipient during that epoch, even before claim.
This cross-sender linkability is inherent to the epoch-scoped binding model and does not affect on-chain observer privacy (G1, G2), because $K_{B,e}$ is never committed on-chain.
Shorter binding epochs reduce the collusion window at the cost of more frequent attestation refreshes.
\end{proposition}

\subsection{Relay Compromise}\label{sec:relay-compromise}

If $\mathcal{R}$'s private database is compromised, the adversary learns the directory mappings
\[
    \mathsf{Norm}(e_B) \mapsto (\IDcom_B, \{u_{B,e}, K_{B,e}, e\}, \texttt{metadata})
\]
and the identifier-to-intent mappings.
This lets the adversary reconstruct the email-to-intent graph and, if public claim events exist, join it to claim destinations and identity commitments.
More seriously, for every retained epoch handle $u_{B,e}$ the adversary can recompute
\[
    H(\texttt{"hfipay:bind"} \| u_{B,e} \| \texttt{intentId}_i)
\]
against public \texttt{intentId} values in the covered time window, retrospectively deanonymizing all intents linked to that epoch.
A compromise therefore exposes all retained history for the compromised epochs, including settled intents, and not only the currently pending graph.

Existing payment infrastructure has the same exposure for privacy metadata: PayPal, Stripe, and traditional banks all maintain private sender-recipient mappings that would expose transaction graphs if breached.
In addition, \HFIPay{} inherits an application-layer availability dependency because $\mathcal{R}$ performs notification and relaying.
Compromise of $\mathcal{R}$ does not, however, grant unilateral claim authority once an intent has committed its blinded binding on-chain, because claims must still satisfy the claim relation $\Phi_{\mathsf{claim}}$ for the corresponding deterministic identity.

Mitigations include:
\begin{itemize}
    \item Encrypting the mapping database at rest with hardware-backed keys.
    \item Rotating binding epochs and deleting settled epoch handles after a bounded retention period.
    \item Purging claimed intent records and stale quote transcripts after a retention period.
    \item Supporting user-visible audit logs for notification, claim, and disclosure actions.
    \item Using attested enrollment records and sender-verifiable quotes in the verified-quote mode so that directory population and quote generation are auditable.
    \item Deriving per-intent bindings through an oblivious PRF or threshold directory so that no single relay node ever holds a complete epoch handle $u_{B,e}$ (future work; see~\cite{voprf}).
    \item Distributing the relay service or directory across multiple independent operators (future work).
\end{itemize}

Third parties who do not obtain the compromised directory still see only pseudonymous intent data on the public chain.

\section{Layered Trust Model}\label{sec:trust}

Identifier-based crypto payments must balance privacy against recoverability and lawful accountability.
\HFIPay{} separates these concerns across three layers and makes the enrollment mode explicit.

\subsection{Three-Layer Architecture}

\begin{table}[H]
\centering
\caption{Layered trust architecture.}
\label{tab:layers}
\begin{tabularx}{\textwidth}{>{\raggedright\arraybackslash}p{2.7cm}>{\raggedright\arraybackslash}X>{\raggedright\arraybackslash}X}
\toprule
\textbf{Layer} & \textbf{Privacy Property} & \textbf{Accountability Mechanism} \\
\midrule
Protocol (on-chain) & Random intents, blinded bindings, quoted asset tuples, and \ZKACE{} claim proofs; no identifier appears on-chain before claim & Verifier logic, refund rules, exact-amount release, and public state transitions \\
\addlinespace
Application / Directory ($\mathcal{R}$) & Private mapping from normalized identifiers to $(\IDcom, u_e, K_e)$ and private quote/intention records & Legal process $\to$ $\mathcal{R}$ $\to$ logs, quote records, and directory state \\
\addlinespace
Identity / Binding ($\mathcal{I}$) & Email or OIDC providers authenticate identifier control; verified-quote mode adds attestations over binding-key commitments & Legal process $\to$ provider or attester $\to$ enrollment records \\
\bottomrule
\end{tabularx}
\end{table}

The trust split differs slightly by deployment mode:
\begin{itemize}
    \item \textbf{Baseline deployment.} The application is trusted both for the one-time binding $\mathsf{Norm}(e_B) \mapsto (\IDcom_B, u_{B,e})$ and for correct per-intent recipient binding before funding.
    \item \textbf{Verified-quote deployment.} The binding layer attests to $(\mathsf{Norm}(e_B), K_{B,e}, e)$ and the sender verifies a quote proof before funding, removing the relay's unilateral recipient-substitution capability.
\end{itemize}

\subsection{Investigation Path}

For a specific investigation, the path is:
\begin{enumerate}
    \item Obtain a target identifier from off-chain sources.
    \item Legal process $\to$ email/OIDC provider or attester $\to$ enrollment identity evidence.
    \item Legal process $\to$ relay/directory $\to$ intent records associated with that identifier.
    \item Intent records $\to$ on-chain intent states and claim events $\to$ fund flows.
\end{enumerate}

The identifier-to-intent mapping stays off the public chain but can be reached through legal process, an audit path that is harder to follow in ordinary pseudonymous DeFi.

\subsection{Comparison with Mixing Protocols}

\begin{table}[H]
\centering
\caption{Comparison with mixing/privacy protocols.}
\label{tab:comparison}
\begin{tabularx}{\textwidth}{lcc>{\raggedright\arraybackslash}X}
\toprule
\textbf{Property} & \textbf{Mixer} & \textbf{Standard DeFi} & \textbf{\HFIPay{}} \\
\midrule
Sender knows recipient? & No & By address & By email \\
Identity anchor & None & None & Email \\
Legal-process path & None & Difficult & Available through provider and relay records \\
On-chain privacy (public observer) & Full & None & Strong pre-claim; after claim, all claims of a recipient are linkable by $\IDcom_B$ in this instantiation \\
\bottomrule
\end{tabularx}
\end{table}

\section{Implementation Sketches}\label{sec:implementation}

We outline compatible realizations on Ethereum-like and Solana-like environments.
These sketches are not a single definitive implementation; they show that the protocol requires only standard chain capabilities: deterministic addresses or state accounts, storage for intent metadata, and a verifier for the chosen \ZKACE{} proof system.
We present implementation compatibility and state layout here, and report measured claim-proof and on-chain gas costs from the reference implementation in Section~\ref{sec:benchmarks}.
End-to-end relay latency and verified-quote proof sizes are not yet benchmarked and remain future work.

\subsection{EVM: Factory, Intent Registry, and Deterministic Deposits}

An EVM-compatible realization can separate concerns into a factory/registry contract and per-intent deposit accounts.

\paragraph{HFIPayFactory.}
The factory serves as both deterministic address oracle and intent registry.
It stores an \texttt{IntentMeta} record keyed by \texttt{intentId}, containing at least the blinded binding $\rho$, the quoted asset descriptor $a$, the quoted amount $v$, the epoch label $e$, the expiry, refund metadata, and the lifecycle status.
It computes deposit addresses via CREATE2 and can deploy EIP-1167-style minimal proxy clones when a claim or refund is executed.
Key entry points are:

\begin{itemize}
    \item \texttt{registerIntent(...)}: records the blinded binding, quoted asset, quoted amount, epoch label, expiry, and refund data before funding.
    \item \texttt{computeAddress(intentId)}: predicts the deposit address off-chain, enabling immediate funding after registration.
    \item \texttt{claimWithProof(...)}: checks the stored blinded binding, quoted asset, epoch label, quoted amount, and the \ZKACE{} proof, then releases exactly the quoted asset and amount to the destination.
    \item \texttt{refund(...)}: after expiry, reconstructs the refund message from the stored intent fields, verifies the revealed $\sigma_A^{\mathsf{ref}}$ against $\gamma_A$'s authorization policy and the stored commitment $h^{\mathsf{ref}}$, and releases exactly the quoted asset and amount to $\gamma_A$.
\end{itemize}

\paragraph{Claim semantics.}
On EVM, the relay cannot change the intended claimant after funding, because \texttt{rho} and the full quoted asset tuple are recorded before the deposit is made.
In the verified-quote deployment, the sender has already checked off-chain that the quoted \texttt{rho} matches the attested recipient binding key, that the quoted address matches the deterministic derivation from \texttt{intentId}, and that the readable registry tuple matches the accepted quote.
The chain therefore verifies the claimant against a pre-committed blinded value and exact asset tuple, not against a relay-supplied identifier at claim time.

Gas cost depends on the chosen verifier and on whether the implementation releases funds directly or via a transient deployed account.
Those details are proof-system and contract-design dependent, so we treat them as implementation-specific rather than protocol-defining.

\subsection{Wallet-Native Resolver API}

A practical deployment can expose the one-time address rail through a resolver API:
\[
    \texttt{POST /api/resolver/hfi},
\]
where \texttt{hfi} denotes a human-friendly identifier.
The request contains the identifier kind, the identifier, the chain, and a coarse asset class; for the minimal EVM rail, the supported asset class is native ETH:
\[
    (\texttt{identifierKind},\texttt{identifier},c,\texttt{asset=native}).
\]
The response contains a one-time protected deposit address, a resolver reference, an expiration time, and display metadata such as ``claim and refund protection.''

Internally, the relay stores a resolver-intent record containing \texttt{resolverRef}, \texttt{intentId}, the normalized identifier in private form, the blinded binding or identifier hash, chain, deposit address, expiry, and lifecycle status.
When the indexer later observes funding to the address, it finalizes the record with \texttt{fundingSender}, \texttt{fundingAmount}, \texttt{fundingTx}, \texttt{claimBefore}, and \texttt{refundAfter}.
The initial reference implementation restricts this rail to native ETH on Base mainnet and returns addresses derived from a configured CREATE2 factory and initialization-code hash.

This API allows a wallet extension or Snap to implement:
\[
    \texttt{alice@example.com} \longrightarrow \texttt{0xOneTimeProtectedDepositAddress}.
\]
The wallet still displays and signs a normal transaction, while \HFIPay{} supplies the private routing, claim, and refund lifecycle behind the one-time address.
The companion application remains useful for Snap installation, claim, status, refund, manual recovery, and the higher-assurance verified-quote rail, but it is no longer required as the primary sender interface for small native wallet sends.

\subsection{Solana: PDA-Based Accounts}

On Solana, the same logic can be implemented with Program Derived Addresses, which replace contract deployment with stateful accounts owned by the program.

Each intent's state is stored in a PDA seeded with \texttt{[b"hfipay", intent\_id]} and contains $\rho$, quoted asset descriptor, quoted amount, epoch label, expiry, refund data, and a status flag.
For SPL token payments, the program additionally creates a PDA-owned token vault $\mathcal{V}_M$ for each supported mint $M$.
The claim instruction requires the \texttt{FUNDED} status and a clock before the stored expiry, verifies the supplied proof against the stored $(\rho, a, e, v, T_{\mathsf{exp}})$ tuple, marks the intent \texttt{CLAIMED}, and then transfers exactly the quoted SOL or SPL asset amount to the claimed destination.
If the target verifier is too expensive to execute monolithically on Solana, a deployment can use a companion verifier program, recursive aggregation, or a proof system chosen for the Solana execution environment.

Unlike the EVM path, no additional contract deployment occurs at claim time; the state PDA and any vault are created during the funding phase.

\subsection{Relay Service}

The relay/directory layer is the sole off-chain component and the only party that holds identifier-to-intent mappings.
Architecturally, it is a computation layer backed by a private database:

\begin{itemize}
    \item \textbf{Directory table}: Maps $\mathsf{Norm}(e_B)$ to the current and unsettled historical binding epochs $(\IDcom_B, u_{B,e}, K_{B,e}, e, \texttt{mode}, \texttt{metadata})$, with metadata such as enrollment time, attestation reference, and rotation state.
    \item \textbf{Intent table}: Stores private intent records containing \texttt{intentId}, $\rho$, asset descriptor, amount, chain, epoch label, refund data, expiry, and status.
    \item \textbf{Intent API}: Accepts $(e_B, a, v, c, \gamma_A, T_{\mathsf{exp}})$, resolves the directory entry privately, computes the blinded binding, and in verified-quote mode returns the full quote tuple together with $(K_{B,e}, T_{\mathsf{bind}}, \tau_B, \pi_i^{\mathsf{quote}})$ and the deterministic deposit address.
    \item \textbf{Transaction relaying}: Submits claim and refund transactions on behalf of users, paying gas or rent so that neither sender nor recipient needs native tokens.
    \item \textbf{Notification}: Informs recipients of pending payments only after funding detection or after acceptance of a non-refundable relay fee, limiting notification spam.
    \item \textbf{Abuse controls}: Enforces sender authentication, rate limits, and quote expiration to reduce unfunded-intent spam.
\end{itemize}

The relay does not hold $\mathsf{REV}_B$, does not custody funds, and cannot unilaterally redirect claims after the blinded binding has been registered on-chain.
Its trust role is therefore narrower than that of a custodian but broader than that of a pure transport relay: in the baseline mode it is trusted for correct one-time directory enrollment and pre-funding recipient binding, while in the verified-quote mode the sender-verifiable quote removes the latter trust assumption.

\subsection{Complexity and Comparative Overhead}

At the protocol level, the extra overhead relative to a naive ``identifier $\to$ address'' mapping is easy to characterize (measured costs follow in Section~\ref{sec:benchmarks}).
The verified-quote deployment adds one off-chain quote proof generation and one sender-side verification per funded intent; both baseline and verified-quote deployments require one claim proof generation by the recipient and one on-chain proof verification at claim time.
The relay API path adds one directory lookup and one intent-registration step, while the sender path adds local address derivation and, in the verified-quote deployment, validation of the registered intent tuple.

\begin{table}[H]
\centering
\caption{Qualitative comparison with simpler identifier-routing approaches.}
\label{tab:complexity}
\begin{tabularx}{\textwidth}{>{\raggedright\arraybackslash}p{3.0cm}>{\raggedright\arraybackslash}X>{\raggedright\arraybackslash}X>{\raggedright\arraybackslash}X}
\toprule
\textbf{Scheme} & \textbf{Extra Off-Chain Work} & \textbf{On-Chain Work} & \textbf{Privacy / Trust Tradeoff} \\
\midrule
Naive direct mapping & None beyond identifier lookup & Ordinary transfer only & No privacy against anyone who knows the identifier \\
\addlinespace
\HFIPay{} baseline & Relay directory lookup; optional sender refund authorization; recipient claim proof generation & Intent registration plus one claim-proof verification & Preserves pre-claim observer privacy, but relay remains trusted for pre-funding recipient binding \\
\addlinespace
\HFIPay{} verified quote & All baseline costs plus one quote proof generation by the relay and one quote verification by the sender & Same on-chain path as baseline & Removes unilateral relay recipient substitution before funding, while keeping identifier resolution off-chain \\
\bottomrule
\end{tabularx}
\end{table}

The byte size of sender-visible quote material is $O(|\tau_B| + |\pi^{\mathsf{quote}}|)$, and the claim transaction size is $O(|\pi_i|)$.
End-to-end latency is dominated by directory lookup, quote-proof generation, and one on-chain registration/funding round before claim.
We report measured claim-proof and on-chain costs below; the quote-proof size and full end-to-end relay latency depend on deployment-specific choices and are left to a future artifact paper.

\subsection{Reference Implementation and Benchmarks}\label{sec:benchmarks}

We report measurements from the reference implementation. The claim relation $\Phi_{\mathsf{claim}}$ is realized in two interchangeable backends: a succinct Groth16/BN254 proof for direct on-chain verification, and a transparent, post-quantum-oriented Circle STARK proof (Stwo, Poseidon2 over Mersenne-31). Proving and verification are measured single-threaded on Apple Silicon in release builds over 30 iterations after 3 warmups (median), with the HFI STARK backend linked against the \ZKACE{} proof engine from \texttt{acechain-io/zk-ace} main at commit \texttt{e118ade5}. The Groth16 backend is part of the \HFIPay{} prototype, and its timing uses the \texttt{zk-ace/dev} deterministic setup for benchmarking only; \ZKACE{}~\cite{zkace} itself measures only a Circle STARK backend. On-chain verifier gas is measured with Foundry. Table~\ref{tab:bench-claim} contrasts the two proof backends, while Table~\ref{tab:bench-gas} separately reports the surrounding EVM intent-operation costs for the attested-binding deployment.

Two caveats apply. First, the prototype claim circuit at the referenced commit implements the identity-commitment and epoch-handle sub-relations of $\Phi_{\mathsf{claim}}$; the full relation of Section~\ref{sec:claim} additionally enforces the salt/domain commitment opening and the blinded-binding linkage $\rho_i = H(\texttt{"hfipay:bind"} \| u \| \texttt{intentId}_i)$, so the reported row and permutation counts should be read as a \emph{lower bound} for the complete relation. In addition, the prototype exposes the epoch handle $u$ as a public input and checks the blinded binding outside the circuit; that departs from Section~\ref{sec:claim}, where $u$ is a private witness, and publishes $u$ at the first claim of an epoch, which breaks enumeration resistance for that epoch. The privacy statements of this paper apply to the relation as specified, not to the prototype. Second, the Stwo-based STARK prototype does not implement zero-knowledge trace masking; its measurements are validity-proof baselines, and the zero-knowledge property assumed of the claim proof system (Section~\ref{sec:model:didp}) must be discharged by adding masking before deployment.

\begin{table}[H]
\centering
\caption{\HFIPay{} claim proof, two backends. Off-chain prove/verify single-threaded on Apple Silicon (release); on-chain verify gas measured with Foundry against the deployed BN254 verifier.}
\label{tab:bench-claim}
\begin{tabular}{@{}lrr@{}}
\toprule
\textbf{Metric} & \textbf{Groth16 / BN254} & \textbf{Circle STARK} \\
\midrule
Prove time (median)        & 36.79\,ms  & 8.08\,ms \\
Verify time, off-chain     & 1.54\,ms   & 0.90\,ms \\
Proof size                 & \textbf{128\,B} & 86{,}213\,B (${\approx}84.2$\,KB) \\
On-chain verify gas        & \textbf{${\approx}241{,}000$} & needs aggregation \\
Deployment role            & direct per-tx & PQ-oriented, per-block agg. \\
\bottomrule
\end{tabular}
\end{table}

The two backends trade proof size against transparency. In the measured implementation, the Groth16 instantiation supports direct EVM verification: a constant $128$-byte proof verified for ${\approx}241\text{k}$ gas, about $1.4\times$ the average gas of the most expensive surrounding intent operation in Table~\ref{tab:bench-gas} (ERC-20 deposit, $167{,}755$). The Circle STARK instantiation is transparent and post-quantum-oriented but produces an ${\approx}84$\,KB proof, which is too large for direct per-transaction submission and is intended for per-block aggregation, where its fast prover and sub-millisecond off-chain verification help. The protocol itself is proof-system-agnostic (Section~\ref{sec:model:didp}): the $84$\,KB transparent proof is a property of one backend, not of \HFIPay{}. The STARK circuit compiles to the compact AIR summarized in Table~\ref{tab:bench-circuit}. The EVM operation gas in Table~\ref{tab:bench-gas} is a separate measurement for the attested-binding deployment's surrounding intent operations; its \texttt{claim} row does not include on-chain proof verification, so deployments that verify Groth16 directly on the claim path should add the verifier cost from Table~\ref{tab:bench-claim}.

\begin{table}[H]
\centering
\caption{\HFIPay{} claim circuit metrics (Circle STARK).}
\label{tab:bench-circuit}
\begin{tabular}{@{}lr@{}}
\toprule
\textbf{Metric} & \textbf{Value} \\
\midrule
Poseidon2 permutations     & 7 \\
Active rows                & 217 \\
Trace length               & 256 ($\log = 8$) \\
Main trace columns         & 73 \\
Preprocessed columns       & 30 \\
AIR constraint expressions & 34 \\
\bottomrule
\end{tabular}
\end{table}

\begin{table}[H]
\centering
\caption{On-chain gas for the EVM attested-binding deployment (Foundry gas reporter). The \texttt{claim} path routes the quoted amount to the active binding minus fee; on-chain ZK-verifier gas depends on the chosen proof system and is not included here.}
\label{tab:bench-gas}
\begin{tabular}{@{}lrr@{}}
\toprule
\textbf{Operation} & \textbf{Avg gas} & \textbf{Median gas} \\
\midrule
\texttt{bind} (epoch binding)        & 56{,}919  & 59{,}428 \\
\texttt{claim} (route to binding)    & 99{,}792  & 115{,}410 \\
\texttt{refund}                      & 52{,}317  & 64{,}351 \\
\texttt{deposit} (native)            & 76{,}920  & 77{,}202 \\
\texttt{deposit} (ERC-20)            & 167{,}755 & 182{,}796 \\
\bottomrule
\end{tabular}
\end{table}

A claim proof is verified directly on EVM for ${\approx}241\text{k}$ gas. By average gas, the surrounding per-intent operations of the attested-binding deployment (bind, deposit, claim routing, and refund) cost on the order of $5$--$18\times 10^4$ gas, with the ERC-20 deposit median at ${\approx}1.83\times 10^5$ gas. A deployment that performs Groth16 verification directly inside the claim path should therefore budget for both the verifier gas and the relevant routing/storage operation gas. The $217$-row STARK claim circuit is smaller than the generic \ZKACE{} circuit in part because the prototype omits the commitment-opening and blinded-binding sub-relations (see the caveats above); the size of the complete relation has not been measured. End-to-end relay latency and verified-quote proof sizes remain to be benchmarked.

\section{Related Work}\label{sec:related}

\paragraph{Stealth Addresses (CryptoNote, ERC-5564).}
{\sloppy
Stealth addresses originate in CryptoNote~\cite{cryptonote} and underlie Monero's one-time output addresses; ERC-5564~\cite{eip5564} ports the construction to EVM chains. These protocols use ECDH between a sender's ephemeral key and the recipient's stealth meta-address to derive one-time receiving addresses.
This achieves strong unlinkability at the protocol layer, but imposes two costs that \HFIPay{} avoids: (i)~the recipient must continuously scan all on-chain transactions to detect payments, and (ii)~the sender must already know the recipient's stealth meta-address, which itself requires a discovery mechanism.
\HFIPay{} shifts discovery to an application-layer relay and uses blinded claim bindings rather than sender-side ECDH, trading protocol-level decentralization for simpler identifier routing and claim-time proofs.\par}

\paragraph{Identifier-routed and registry-based payment systems.}
Several deployed systems target adjacent points in the design space.
ZK Email's Email Wallet~\cite{zkemail} realizes send-to-email on Ethereum by proving DKIM-signed email receipt in zero knowledge, and is the closest deployed relative of \HFIPay{}; it replaces the relay/issuer trust of this work with trust in DKIM keys and an on-chain registry, and publishes salted identifier commitments on-chain.
zkLogin~\cite{zklogin} lets OAuth identities control chain accounts via ZK proofs over OIDC tokens, providing identifier-addressed payments in an issuer-based model.
BIP-352 Silent Payments~\cite{bip352} achieves reusable off-chain payment identifiers with one-time on-chain addresses on Bitcoin via sender-side ECDH and recipient scanning---the same no-reusable-tag goal as \HFIPay{}, without a relay but with scanning cost.
ERC-6538~\cite{erc6538} standardizes an on-chain stealth meta-address registry (a public identifier-to-key mapping), and Umbra~\cite{umbra} is a deployed EVM stealth-payment system with such a registry; both make the identifier-to-key binding public, which \HFIPay{} deliberately avoids.
Key-transparency systems such as CONIKS~\cite{coniks} address the complementary problem of auditing identifier-to-key directories and could harden this work's enrollment layer; oblivious PRFs~\cite{voprf} could reduce the directory's visibility into queries.

\paragraph{ENS and Email-Based Naming.}
The Ethereum Name Service~\cite{ens} maps human-readable names to blockchain addresses, solving the UX problem of opaque hex strings.
However, the mapping is fully public: anyone who knows an ENS name can inspect the owner's balances and transaction history.
\HFIPay{} addresses the same usability problem while keeping the identifier-to-intent and identifier-to-address mappings private.
ENS provides public, persistent naming; \HFIPay{} provides private, relay-mediated routing.

\paragraph{Custodial ``Send to Email'' Systems.}
Centralized exchanges and payment apps commonly offer send-to-email or send-to-phone workflows by keeping custody of funds and resolving identifiers inside a private operator database.
Those systems are operationally close to \HFIPay{} in that they rely on a relay-like directory, but they differ in trust placement: the operator ultimately decides both routing and release.
\HFIPay{} instead keeps the identifier mapping private while moving post-funding authorization to chain-verifiable proofs, narrowing the relay's authority in the verified-quote mode.

\paragraph{\ZKACE.}
\ZKACE{}~\cite{zkace} authorizes an action on a high-throughput ledger by a zero-knowledge proof of an opening of a registered identity commitment, bound to a typed action digest, a public nonce, and a replay value, and places this proof at account birth. \HFIPay{} follows this proof pattern but uses its own salted commitment and replay rule.
\HFIPay{} uses \ZKACE{} not for sender-side routing but for recipient-side claim authorization: once the relay has privately resolved an identifier, furnished a sender-verifiable quote, and committed the blinded binding $\rho_i$ together with the quoted asset tuple, the recipient can prove on-chain that the same committed identity authorizes release of that specific intent to a destination address.
This composition removes claim correctness from the relay's discretion after intent registration.

\paragraph{Cross-device recovery.}
Recipient recovery (Section~\ref{sec:binding-modes}) re-derives both $\IDcom_B$ and the claim-binding handle; because the funded intent stores the coarse epoch label $e_i$, it does not need a live relay to learn which epoch to re-derive.
It does not solve sender-side lookup of a recipient identity, which remains the job of the relay's private directory.

\paragraph{Account Abstraction (ERC-4337).}
Account abstraction~\cite{erc4337} enables smart contract wallets with flexible authentication, gas sponsorship, and batched operations.
\HFIPay{} addresses the \emph{discovery and routing} problem; account abstraction addresses execution and sponsorship once the claimant or destination is known.
An EVM implementation can therefore expose \HFIPay{} claims through an ERC-4337-compatible smart-account interface if desired.

\paragraph{Shielded and mixing-based payments.}
A large body of work hides payment data on-chain. Zerocash~\cite{zerocash} builds a decentralized anonymous payment system in which a shielded pool, commitments, and zero-knowledge spend proofs conceal sender, recipient, and amount; Zether~\cite{zether} adds confidential (amount-hiding) and anonymous transfers to account-based smart-contract chains; and mixing protocols such as Tornado Cash~\cite{tornado} break sender-recipient links by pooling deposits into a common anonymity set with zero-knowledge withdrawals.
\HFIPay{} targets a narrower privacy goal. It does \emph{not} hide amounts or asset types, which are public on-chain (Table~\ref{tab:visibility}), and it does \emph{not} commingle funds into a shared pool. Its only privacy objective is to keep the mapping from a human-friendly \emph{identifier} to on-chain intents private before claim, while preserving an off-chain identifier anchor for routing and lawful accountability. Consequently \HFIPay{} avoids the anonymity-set and shielded-pool machinery of those systems, at the cost of leaking the payment amount and a coarse epoch label.
This identifier-routing-only stance is closer to the privacy/compliance balance of Privacy Pools~\cite{privacypools} than to full shielding: like that line of work, \HFIPay{} keeps a deliberate accountability path (Section~\ref{sec:trust}) rather than maximizing anonymity.

\section{Conclusion}\label{sec:conclusion}

\HFIPay{} composes four layers: private identifier routing, sender-verifiable quote binding, blinded intent commitment, and on-chain zero-knowledge claim authorization.

Before claim, the chain reveals random intent identifiers, deposit addresses, and one-time blinded bindings, but not human-friendly identifiers or reusable recipient tags.
After claim, linkability depends on what public identity commitment and destination information is reused; in the instantiation presented here, the explicit $\IDcom_B$ argument makes all claims of a recipient linkable.

The relay's trust role is narrower than a custodian's.
In the baseline deployment, the application is still trusted for correct one-time enrollment and for correct pre-funding recipient binding.
In the verified-quote deployment, attested off-chain quotes remove that unilateral pre-funding substitution power from the relay while preserving pre-claim on-chain privacy.
In either case, once the intent's blinded binding has been committed on-chain, the relay can no longer unilaterally redirect the claim.

\subsection{Future Work}

Several directions remain open:

\begin{itemize}
    \item \textbf{Quote-proof optimization.} Compressing or aggregating sender-verifiable quote proofs would reduce latency and make the verified-quote deployment cheaper to operate at scale.
    \item \textbf{Private directory federation.} Distributing or federating the private directory across independent operators could reduce single-point compromise risk while preserving private lookup.
    \item \textbf{Proof aggregation.} Batched or recursive verification of \ZKACE{} claim proofs would reduce verifier cost for high-throughput payment hubs.
    \item \textbf{Lazy first-receipt onboarding.} Formalizing a strong first-receipt enrollment flow for previously unregistered recipients remains an important deployment problem.
    \item \textbf{Unlinkable claims.} Replacing the explicit $\IDcom_B$ claim argument with a set-membership relation over the registered-commitment set, removing the stable post-claim tag identified in Section~\ref{sec:privacy}.
    \item \textbf{Handle-free directory.} An OPRF- or threshold-based directory in which the relay computes blinded bindings without ever holding complete epoch handles, shrinking the retrospective blast radius of a directory compromise.
    \item \textbf{Full claim-relation artifact.} Extending the prototype claim circuit to the complete relation (a)--(d) of Section~\ref{sec:claim}, adding STARK trace masking, and re-benchmarking; the current measurements are lower bounds (Section~\ref{sec:benchmarks}).
\end{itemize}

\bibliographystyle{plain}

\begin{thebibliography}{99}

\bibitem{acegf}
J.~S.~Wang.
\newblock {ACE-GF}: A Generative Framework for Atomic Cryptographic Entities.
\newblock arXiv preprint arXiv:2511.20505, 2025.

\bibitem{zkace}
J.~S.~Wang.
\newblock {ZK-ACE}: Authorization from Registered Identity Commitments.
\newblock arXiv preprint arXiv:2603.07974, 2026.

\bibitem{poseidon}
L.~Grassi et~al.
\newblock Poseidon: A new hash function for zero-knowledge proof systems.
\newblock In \emph{USENIX Security}, 2021.

\bibitem{poseidon2}
L.~Grassi et~al.
\newblock Poseidon2: A Faster Meta-Hash Function for SNARKs and STARKs.
\newblock \emph{IACR ePrint} 2023/323, 2023.

\bibitem{stwo}
StarkWare Industries.
\newblock Stwo: A Circle STARK prover.
\newblock \url{https://github.com/starkware-libs/stwo}, 2024--2026.

\bibitem{eip5564}
T.~Wahrst\"atter, M.~Solomon, B.~DiFrancesco, and V.~Buterin.
\newblock {ERC-5564}: Stealth Addresses.
\newblock Ethereum Improvement Proposals, 2022.

\bibitem{ens}
N.~Johnson.
\newblock {ERC-137}: Ethereum Domain Name Service --- Specification.
\newblock Ethereum Improvement Proposals, 2016.

\bibitem{erc4337}
V.~Buterin, Y.~Weiss, D.~Tirosh, S.~Nacson, A.~Forshtat, K.~Gazso, and T.~Hess.
\newblock {ERC-4337}: Account Abstraction Using Alt Mempool.
\newblock Ethereum Improvement Proposals, 2021.

\bibitem{tornado}
A.~Pertsev, R.~Semenov, and R.~Storm.
\newblock Tornado Cash Privacy Solution.
\newblock White paper, 2019.

\bibitem{zerocash}
E.~Ben-Sasson, A.~Chiesa, C.~Garman, M.~Green, I.~Miers, E.~Tromer, and M.~Virza.
\newblock Zerocash: Decentralized anonymous payments from Bitcoin.
\newblock In \emph{IEEE Symposium on Security and Privacy (S\&P)}, pp.~459--474, 2014.

\bibitem{cryptonote}
N.~van Saberhagen.
\newblock CryptoNote v2.0.
\newblock White paper, 2013.

\bibitem{privacypools}
V.~Buterin, J.~Illum, M.~Nadler, F.~Sch\"ar, and A.~Soleimani.
\newblock Blockchain privacy and regulatory compliance: Towards a practical equilibrium.
\newblock \emph{Blockchain: Research and Applications}, 5(1), 2024.

\bibitem{zether}
B.~B\"unz, S.~Agrawal, M.~Zamani, and D.~Boneh.
\newblock Zether: Towards privacy in a smart contract world.
\newblock In \emph{Financial Cryptography and Data Security (FC)}, pp.~423--443, 2020.

\bibitem{bellare-goldwasser}
M.~Bellare and S.~Goldwasser.
\newblock New paradigms for digital signatures and message authentication based on non-interactive zero knowledge proofs.
\newblock In \emph{CRYPTO 1989}, LNCS 435, pp.~194--211. Springer, 1990.

\bibitem{gmr-sig}
S.~Goldwasser, S.~Micali, and R.~L.~Rivest.
\newblock A digital signature scheme secure against adaptive chosen-message attacks.
\newblock \emph{SIAM J. Comput.}, 17(2):281--308, 1988.

\bibitem{zkemail}
ZK Email project.
\newblock Email Wallet: send assets to email addresses via ZK proofs of DKIM-signed email.
\newblock \url{https://zk.email/}, 2023--2026.

\bibitem{zklogin}
F.~Baldimtsi, K.~K.~Chalkias, et~al.
\newblock zkLogin: Privacy-preserving blockchain authentication with existing credentials.
\newblock arXiv:2401.11735, 2024.

\bibitem{bip352}
J.~Baker (josibake) and R.~Somsen.
\newblock BIP-352: Silent Payments.
\newblock Bitcoin Improvement Proposal, 2024.

\bibitem{erc6538}
M.~Solomon, T.~Wahrst\"atter, B.~DiFrancesco, and V.~Buterin.
\newblock {ERC-6538}: Stealth Meta-Address Registry.
\newblock Ethereum Improvement Proposals, 2023.

\bibitem{umbra}
ScopeLift.
\newblock Umbra: privacy-preserving stealth payments on EVM chains.
\newblock \url{https://app.umbra.cash/}, 2021--2026.

\bibitem{coniks}
M.~S.~Melara, A.~Blankstein, J.~Bonneau, E.~W.~Felten, and M.~J.~Freedman.
\newblock CONIKS: Bringing key transparency to end users.
\newblock In \emph{USENIX Security}, 2015.

\bibitem{voprf}
A.~Davidson, A.~Faz-Hern\'andez, N.~Sullivan, and C.~A.~Wood.
\newblock Oblivious pseudorandom functions (OPRFs) using prime-order groups.
\newblock RFC 9497, 2023.

\end{thebibliography}

\end{document}